\documentclass[smallextended]{svjour3}       % onecolumn (second format)
\smartqed  % flush right qed marks, e.g. at end of proof
\usepackage{graphicx}
\usepackage{amsmath,amssymb,amsfonts}
\usepackage{color}
\usepackage{enumerate}
\usepackage{algorithmic,algorithm}
\usepackage{graphics, graphicx}

\spnewtheorem*{problem1}{Problem}{\bfseries}{\rmfamily}

\newcommand{\lca}{\ensuremath{\operatorname{lca}}}

\newcommand{\Mo}{\ensuremath{M^\odot}}
\newcommand{\Mop}{\ensuremath{\widetilde{M}^\odot}}
\renewcommand{\varnothing}{\odot}
\newcommand{\de}{\ensuremath{\delta}}

\renewcommand{\hat}{\widehat}

 \journalname{}%	{J Comb Opt}
\begin{document}

\title{On Tree Representations of Relations and Graphs
		\thanks{Parts of this paper were presented at the 
					21st Annual International Computing and Combinatorics Conference (COCOON 2015), 
  					August 4-6, 2015, Beijing, China \cite{HW:15}. 
					This work was funded by the German Research Foundation
					(DFG) (Proj. No. MI439/14-1).}
%about the article that should go on the front page should be
%placed here. General acknowledgments should be placed at the end of the article.}
}
\subtitle{Symbolic Ultrametrics and Cograph Edge Decompositions}

\titlerunning{Tree Representations of Graphs}        % if too long for running head

\author{Marc Hellmuth         \and
        Nicolas Wieseke %etc.
}

%\authorrunning{Short form of author list} % if too long for running head

\institute{M. Hellmuth \at
			 \emph{Department of Mathematics and Computer Science, University of Greifswald, 
			 Walther-Rathenau-Strasse 47,
			 D-17487 Greifswald, 
			 Germany}
			\\
			Center for Bioinformatics, Saarland University, Building E 2.1, D-66041
			Saarbr{\"u}cken, Germany \\
            \email{mhellmuth@mailbox.org}             \\
%             \emph{Present address:} of F. Author  %  if needed
           \and
           N. Wieseke \at
           Parallel Computing and Complex Systems Group, Department of
			Computer Science, Leipzig University, Augustusplatz 10, D-04109
			Leipzig, Germany\\
			\email{wieseke@informatik.uni-leipzig.de}
}

\date{Received: date / Accepted: date}
% The correct dates will be entered by the editor

\maketitle

\begin{abstract}
Tree representations of (sets of) symmetric binary relations, or equivalently
edge-colored undirected graphs, are of central interest, e.g.\ in phylogenomics.
In this context symbolic ultrametrics play a crucial role. Symbolic ultrametrics
define an edge-colored complete graph that allows to represent the topology of
this graph as a vertex-colored tree. Here, we are interested in the structure
and the complexity of certain combinatorial problems resulting from
considerations based on symbolic ultrametrics, and on algorithms to solve them.

This includes, the characterization of symbolic ultrametrics that additionally 
distinguishes between edges and non-edges of \emph{arbitrary} edge-colored 
graphs $G$ and thus, yielding a tree representation of $G$, by means of so-called cographs. 
Moreover, we address the problem of finding ``closest'' symbolic ultrametrics and
show the NP-completeness of the three problems: symbolic ultrametric editing, completion and deletion.
Finally, as not all graphs are cographs, and hence, don't have a tree representation, 
we ask, furthermore, what is the minimum number of cotrees needed to represent the topology of 
an arbitrary non-cograph $G$. This is equivalent to find
an optimal cograph edge $k$-decomposition $\{E_1,\dots,E_k\}$ of $E$ so that each
subgraph $(V,E_i)$ of $G$ is a cograph. We investigate this problem in full detail,
resulting in several new open problems, and NP-hardness results. 

For all optimization problems proven to be NP-hard we will provide integer
linear program (ILP) formulations to efficiently solve them.
%This includes, the characterization of 
%We first ask under what conditions it is possible to find a symbolic ultrametric that additionally 
%distinguishes between edges and non-edges of \emph{arbitrary} edge-colored 
%graphs $G$ and thus, yielding a tree representation of $G$. We prove that 
%such a symbolic ultrametric can only be defined for $G$ if and
%only if $G$ is a cograph, that are graphs that do not contain an induced path on any
%subset of four vertices.
%Moreover, we address the problem of finding ``closest'' symbolic ultrametrics. 
%We show the NP-completeness of the three problems: symbolic ultrametric editing, completion and deletion, 
%and provide integer linear program (ILP) formulations to efficiently solve
%these problems. 
%As not all graphs are cographs, and hence, don't have a tree representation, 
%we ask, furthermore, what is the minimum number of cotrees needed to represent the topology of 
%an arbitrary non-cograph $G$ which is equivalent to find
%an optimal cograph edge $k$-decomposition $\{E_1,\dots,E_k\}$ of $E$ so that each
%subgraph $(V,E_i)$ of $G$ is a cograph. We investigate this problem in detail,
%resulting in a couple of conjectures,  and additionaly 
%provide NP-hardness results as well ILP formulations to solve this problem. 
%a
%Anis derived and it is 
%shown that determining whether a graph has a cograph 2-decomposition, resp., 2-partition is NP-complete.
%
% \PACS{PACS code1 \and PACS code2 \and more}
% \subclass{MSC code1 \and MSC code2 \and more}
\end{abstract}

\sloppy
\section{Introduction}

Tree representations of relations between certain objects lie at the heart
of many problems, in particular, in phylogenomic studies
\cite{HHH+13,HLS+14,LE:15}. 
Phylogenetic Reconstructions are concerned with the study of the evolutionary
history of groups of systematic biological units, e.g. genes or species. The
objective is the assembling of so-called phylogenetic trees or networks that
represent a hypothesis about the evolutionary ancestry of a set of genes,
species or other taxa.

Consider a symmetric map $\de\colon V\times V\to \Mo$
that assigns to each pair $(x,y)$ a symbol or color $m\in \Mo$. The question
then arises whether it is possible to determine a rooted tree $T$ with a
vertex-labeling $t$ so that the lowest common ancestor $\lca(x,y)$ of distinct
leaves $x$ and $y$ in $T$ is labeled with $m\in \Mo$ if and only if $\de(x,y)=m$. Such
a tree is then called \emph{symbolic representation} of $\de$. In phylogenomics,
such maps $\de$ can be interpreted as an assignment of evolutionary
relationships between two genes and the symbolic representation of
such relations then reflect the evolutionary history together with the
respective events that happened when two genes diverged. 
It has recently be shown, that in theory \cite{HHH+12} and in
practice \cite{HLS+14} it is even possible 
to reconstruct the evolutionary history of species, 
where the genes have been taken from, whenever
the  symbolic representation is known.

The problem of finding such symbolic representations $(T,t)$ has been first
addressed by B{\"o}cker and Dress \cite{BD98} in a mathematical context. 
The authors showed, that there is
a symbolic representation $(T,t)$ of $\de$ if and only if the map $\delta$
fulfills the properties of a so-called symbolic ultrametric \cite{BD98}. 
Clearly, any such map $\de\colon V\times V\to \Mo$ 
is equivalent to a set of disjoint symmetric binary relations $\{R_m\mid m\in \Mo\}$
with $(x,y)\in R_{\de(x,y)}$ or 
an edge-colored complete graph $K_{|V|}=(V, {{V}\choose{2}})$
so that each edge $[x,y]$ obtains the color $\de(x,y)\in \Mo$.
In \cite{HHH+13} a characterization of symbolic ultrametrics
in graph theoretical terms have been given. 
It has been shown that there is a symbolic representation $(T,t)$
of such an edge-colored graph $K_{|V|}$ if the edges of each cycle of length $3$
have at most two colors and each mono-chromatic subgraph $G_m$, i.e., subgraphs 
that consist of all the edges having a fixed color $m$, are so-called cographs.
Cographs are characterized by the absence of induced
paths $P_4$ on four vertices, although there are a number of equivalent characterizations
of cographs (see e.g.\ \cite{Brandstaedt:99} for a survey).
Moreover, Lerchs \cite{lerchs71,lerchs72}
showed that each cograph $G=(V,E)$ is associated with a unique rooted tree $T(G)$,  
called \emph{cotree}.

In this contribution we address several combinatorial problems that
are concerned with symbolic ultrametrics and tree representations of
arbitrary, possibly edge-colored graphs. 

We first investigate in Section \ref{sec:suc}, under what conditions it is possible to find a symbolic
ultrametric for arbitrary graphs $G$ so that edges and non-edges of $G$ can be
distinguished. In other words, we ask for an edge-coloring of $G$ so that edges
and non-edges always obtain different colors and this coloring satisfies the conditions 
of a symbolic ultrametric. If such a coloring is known for $G$, then
one can immediately display the topology of $G$ as a tree via a symbolic representation
$(T,t)$. It does not come as a big surprise, when we prove that 
such a symbolic ultrametric can only be defined for
$G$ if and only if $G$ is already a cograph.  This, in particular, establishes
another new characterization of cographs.
As a consequence we can infer that any symbolic representation 
$(T,t)$ of a cograph $G$ is a so-called refinement of its cotree. 

In practice, however, symmetric maps $d\colon V\times V\to \Mo$ represent often only estimates of the true relationship $\de$
between the investigated objects, e.g., genes \cite{Lechner:11a,Lechner:14}.  
Thus, in general such estimates $d$ will not be a symbolic ultrametric. 
Hence, there is a great interest in optimally editing $d$ to a
symbolic ultrametric $\de$, i.e., finding a minimum number of 
changes of the assignment $d(x,y)\in \Mo$ to pairs $(x,y)$ so that there
is a symbolic representation of the resulting map $\de$ \cite{HLS+14}. 
So-far, the complexity of this problem has been unknown. 
In Section \ref{sec:sue}, we show that
(the decision version of) this problem, called \textsc{Symbolic Ultrametric Editing}, is NP-complete. 
Additionally, we show that the problems \textsc{Symbolic Ultrametric Completion}
and \textsc{Symbolic Ultrametric Deletion} are NP-complete and 
provide integer linear program (ILP) formulations
in order to efficiently solve the latter three problems in future work.

A further combinatorial problem we consider in Section \ref{sec:crd} is
motivated by the results established in Section \ref{sec:suc} where we have characterized
graphs for which one can find symbolic representations by means of cographs.
However, not all graphs are cographs and thus, don't have such a tree
representation. Therefore, we ask for the minimum number of cotrees that are
needed to represent the structure of a given graph $G=(V,E)$ in an unambiguous
way. As it will turn out, this problem is equivalent to find a decomposition
$\Pi = \{E_1,\dots, E_k\}$ of $E$ (i.e., the elements of $\Pi$ need not
necessarily be disjoint) for the least integer $k$, so that each subgraph $G_i =
(V,E_i)$,\ $1\leq i\leq k$ is a cograph. Such a decomposition is called cograph
edge $k$-decomposition, or cograph $k$-decomposition, for short. If the elements
of $\Pi$ are in addition pairwise disjoint, we call $\Pi$ a cograph
$k$-partition. We show that the number of such optimal cograph $k$-decomposition,
resp., partitions on a graph can grow exponentially in the number of vertices.
Moreover, non-trivial  upper bounds for the integer $k$ such that there is a 
cograph $k$-decomposition, resp., partition are derived and 
polynomial-time algorithms 
to compute $\Pi$ with $|\Pi|\leq \Delta+1$, where $\Delta$ denotes the maximum
number of edges a vertex is contained in, are provided.
Furthermore, we will prove that finding the least integer $k\geq 2$ so that $G$ has a
cograph $k$-decomposition or a cograph $k$-partition is an NP-hard problem.
In order to attack this problem in future work, we derive ILP 
formulations to solve this problem efficiently. These
findings complement results known about so-called cograph vertex partitions
\cite{Achlioptas199721,GN:02,DMO:14,MSC-Zheng:14}.

\section{Essential Definitions}

\noindent
\emph{Graph.}
In what follows, we consider undirected simple graphs $G=(V,E)$ with vertex
set $V(G) = V$ and edge set $E(G) = E\subseteq \binom{V}{2}$. The \emph{complement graph} $G^c=(V,E^c)$ of
$G=(V,E)$, has edge set $E^c= \binom{V}{2}\setminus E$.
The graph $K_{|V|} = (V,E)$ with $E=\binom{V}{2}$ is called \emph{complete graph}. 
A graph $H=(W,F)$ is an \emph{induced subgraph} of $G=(V,E)$, if $W\subseteq V$
and all edges $[x,y]\in E$ with $x,y\in W$ are contained in $F$.
The \emph{degree} $\deg(v)=|\{e\in E\mid v\in e\}|$ of a vertex $v\in V$
is defined as the number of edges that contain $v$. The maximum degree 
of a graph is denoted with $\Delta$.

%\vspace{0.4em}
\bigskip
\noindent
\emph{Rooted Tree.}
A connected graph $T$ is a \emph{tree}, if $T$ does not contain cycles. A
vertex of a tree $T$ of degree one is called a \emph{leaf} of $T$ and all other
vertices of $T$ are called \emph{inner} vertices. 
The set of inner vertices of $T$ is denoted by $V^0$ and with $E^0$ 
we denote the set of \emph{inner} edges, that are the edges in $E$ where 
both of its end vertices are inner vertices. 
A \emph{rooted tree}
$T=(V,E)$ is a tree that contains a distinguished vertex $\rho_T\in V$
called the \emph{root}. 
The first inner vertex  $\lca_{T}(x,y)$ that lies on both unique paths from 
distinct leaves $x$, resp., $y$ to the root, is called \emph{lowest
common ancestor of $x$ and $y$}.
If there is no danger of ambiguity, we will write $\lca(x,y)$ rather then	
$\lca_T(x,y)$. 

%\vspace{0.4em}
\bigskip
\noindent
\emph{Symbolic Ultrametric and Symbolic Representation.}
In what follows, the set $M$ will always denote a non-empty finite
set, the symbol $\varnothing$ will always denote a special
element not contained in $M$, and $\Mo := M\cup\{\varnothing\}$. 
Now, suppose $X$ is an arbitrary non-empty set and 
$\delta:X \times X \to \Mo$ a map.
We call $\delta$ a \emph{symbolic ultrametric}
if it satisfies the following conditions:
\begin{enumerate}
  \item[(U0)] $\delta(x,y)=\varnothing$ if and only if $x=y$; %\vspace{-0.2cm}
  \item[(U1)] $\delta(x,y) = \delta(y,x)$ for all $x,y \in X$, i.e.\ $\delta$ is
				symmetric; %\vspace{-0.2cm}
  \item[(U2)] $|\{\delta(x,y), \delta(x,z),\delta(y,z)\}| \le 2$
    for all $x,y,z \in X$; and %\vspace{-0.2cm}
  \item[(U3)] there exists no subset $\{x,y,u,v\}\in {X\choose 4}$ such that
%    \begin{equation}
  $    \delta(x,y)= \delta(y,u)  = \delta(u,v) \neq 
      \delta(y,v)= \delta(x,v)  = \delta(x,u).$
 %   \end{equation}
\end{enumerate}
Now, suppose that $T=(V,E)$ is a rooted tree with leaf set $X$ 
and that $t:V\to \Mo$ is a map 
such that $t(x)=\varnothing$ for all $x\in X$.
To the pair $(T,t)$ we associate the map 
$d_{(T,t)}$ on $X \times X$ by setting, for all $x,y \in X$,
\begin{equation*}
d_{(T,t)}: X \times X \to \Mo;  d_{(T,t)}(x,y) = t(\lca_{T}(x,y)).
\end{equation*}
%
%Clearly this map is symmetric and satisfies (U0). 
We call the pair
$(T,t)$ a \emph{symbolic representation} of a map $\delta:X \times X \to \Mo$, 
if $\delta(x,y)=d_{(T,t)}(x,y)$ holds for all $x,y\in X$. 
For a subset $W\subseteq X\times X$ we denote with 
$\de(W) = \{m\in\Mo\mid\exists x,y\in W \text{\ s.t\ } \de(x,y)=m\}$
the set of images of the elements contained in $W$.
%restriction of $\de$ to the set $W$.

	We say that $(T, t)$ and $(T',t')$ are isomorphic if 
	$T$ and $T'$ are isomorphic via a map
   $\varphi:V(T)\to V(T')$ such that $t'(\varphi(v))= t(v)$ holds 
	for all $v\in V(T)$.

%\TODO{discriminating symbolic representation}

%\vspace{0.4em}
\bigskip
\noindent
\emph{Cographs and Cotrees.} 
Complement-reducible graph, cographs for short, 
are defined as the class of graphs formed from a single vertex under 
the closure of the operations of union and complementation, namely: (i) 
a single-vertex graph is a cograph;
(ii) the disjoint union of cographs is a cograph; 
(iii) the complement of a cograph is a cograph. 
Alternatively, a cograph can be defined as  a 
$P_4$-free graph
(i.e.\ a graph such that no four vertices induce a subgraph that is a path
of length 3). A number of equivalent characterizations
of cographs can be found in \cite{Brandstaedt:99}.
It is well-known in the literature concerning cographs that,
to any cograph $G=(V',E')$, one can associate a canonical \emph{cotree}
$T(G)=(V,E)$. This is a rooted tree, leaf set $V\setminus V^0$ equal to the
vertex set $V'$ of $G$ and inner vertices that represent so-called "join"
and "union" operations together with a labeling map $t:V^0\to \{0,1\}$ such
$[x,y]\in E'$ if and only if $t(\lca(x,y)) = 1$, and
$t(v)\neq t(w_i)$ for all $v\in V^0$ and all children $w_1,\ldots, w_k\in V^0$
of $v$, (cf. \cite{Corneil:81}). We will call the pair $(T,t)$
\emph{cotree representation} of $G$.

%\vspace{0.4em}
\bigskip
\noindent
\emph{Cograph $k$-Decomposition and Partition, and Cotree Representation.}
Let  $G=(V,E)$ be an arbitrary graph. 
A decomposition $\Pi=\{E_1,\dots E_k\}$ of $E$ is a called \emph{(cograph) $k$-decomposition}, 
if each subgraph $G_i = (V,E_i)$,\ $1\leq i\leq k$ of $G$ is a cograph.
We call $\Pi$ a \emph{(cograph) $k$-partition} if $E_i\cap E_j=\emptyset$, for all 
distinct $i,j\in \{1,\dots,k\}$.
A $k$-decomposition $\Pi$ is called \emph{optimal}, if $\Pi$ has the least number $k$
of elements among all cograph decompositions of $G$. Clearly, for a cograph only
$k$-decompositions with $k=1$ are optimal.
A $k$-decomposition $\Pi=\{E_1,\dots E_k\}$ is \emph{coarsest}, if no elements
of $\Pi$ can be unified, so that the resulting decomposition is a 
cograph $l$-decomposition, with $l< k$. In other words, $\Pi$ is 
coarsest, if for all subsets $I\subseteq \{1,\dots,k\}$ with $|I|>1$
it holds that $(V, \cup_{i\in I} E_i)$ is not a cograph. 
Thus, every optimal $k$-decomposition is also always a coarsest one. 

A graph $G = (V,E)$ is \emph{represented by a set of cotrees} $\mathbb T = \{T_1,\dots,T_k\}$, 
each $T_i$ with leaf set $V$, if and only if for each edge $[x,y]\in E$ there is a tree $T_i\in \mathbb T$
with $t(\lca_{T_i}(x,y)) =1$.

%\vspace{0.4em}
\bigskip
\noindent
\emph{The Cartesian (Graph) Product}
$G\Box H$ has vertex set $V(G\Box H)=V(G)\times
V(H)$; two vertices $(g_1,h_1)$, $(g_2,h_2)$ are adjacent in $G\Box H$ if
$[g_1,g_2]\in E(G)$ and $h_1=h_2$, or $[h_1,h_2]\in E(H)$ and $g_1 =
g_2$. It is well-known that the Cartesian product is associative, 
commutative and that the single vertex graph $K_1$ serves as unit element
\cite{HIK-13,Hammack:11a}. Thus, the product $\Box_{i=1}^n G_i$ 
of arbitrary many factors $G_1,\dots, G_n$ is well-defined. 
For a given product $\Box_{i=1}^n G_i$, we define 
%assosciativ, commutativ ..
the $G_i$-layer $G_i^w$ of $G$ (through vertex $w$ that has coordinates $(w_1,\dots,w_n)$) 
as the induced subgraph with vertex set
$V(G_i^w)=\{v = (v_1,\dots,v_n)\in \times_{i=1}^n V(G_i) \mid v_j=w_j, \text{ for all } j\neq i\}$. 
%It is
%isomorphic to $G_i$.
Note, $G_i^w$ is isomorphic to $G_i$ for all $1\leq i\leq n$, $w\in V(\Box_{i=1}^n G_i)$. 
The \emph{$n$-dimensional hypercube} $Q_n$ or \emph{$n$-cube}, for short,  is the Cartesian product
$\Box_{i=1}^n K_2$. %, where $K_2$ denotes the complete graph on $2$ vertices.

\section{Symbolic Ultrametrics and Cographs}
\label{sec:suc}

Symbolic ultrametrics and respective representations as event-labeled
trees, have been first characterized by B{\"o}cker and Dress \cite{BD98}.

\begin{theorem}[\cite{BD98,HHH+13}]\label{bd}
Suppose $\delta: V \times V\to\Mo$ is a map. Then there is a 
%discriminating 
symbolic representation of $\delta$ 
if and only if $\delta$ is a symbolic ultrametric. 
Furthermore, this representation 
can be computed in polynomial time. 
\end{theorem}

Now, let $\delta:V \times V \to \Mo$ be a map satisfying 
Properties (U0) and (U1). 
Clearly, the map $\delta$ can be considered as an edge coloring of 
a complete graph $K_{|V|}$, where each edge $[x,y]$ obtains
color $\delta(x,y)$.
For each
fixed $m\in M$, we define the undirected graph $G_m:=G_m(\delta)=(V,E_m)$ 
with edge set
\begin{equation}
E_m =  
\left\{ [x,y] \mid \delta(x,y)=m, \,\, x,y \in V \right\}.
\end{equation}
Hence, $G_m$ denotes the subgraph of the edge-colored graph 
$K_{|V|}$, that contains all edges colored with $m\in M$. 
The following result establishes the connection between
symbolic ultrametrics and cographs.

\begin{theorem}[\cite{HHH+13}]\label{thm:cograph}
  Let $\delta:V\times V\to \Mo$ be a map satisfying Properties 
  (U0) and (U1). Then 
  $\delta$ is a symbolic ultrametric if and only if 
  \begin{itemize}
  \item[(U2')] For all $\{x,y,z\}\in\binom{V}{3}$ there is an $m\in M$ 
	such that
    $E_m$ contains (at least) two of the three edges $[x,y]$, $[x,z]$, and
    $[y,z]$. In other words, for each triangle induced by $x,y$ and $z$, the edges have at most $2$ different colors 
  \item[(U3')] $G_m$ is a cograph for all $m\in M$. 
  \end{itemize}
\end{theorem}

Assume now, we have given an arbitrary none edge-colored
graph $G=(V,E)$ and we want to represent the topology of $G$ 
as a tree. The following question then arises:

\smallskip
\emph{Under which conditions is it possible to define an coloring on the edges
\emph{and} non-edges of $G$, so that edges $e\in E$ obtain a different color
then the non-edges $e\in  E^c$ of $G$
and, in particular, so that the resulting map $\delta$ is a symbolic ultrametric?}
\smallskip

%\end{enumerate}
In other words, we ask for an edge-coloring of $G$ so that there is a 
tree $(T,t)$ with $t(\lca_T(x,y))=m$ if and only if the (non)edge
$[x,y]$ obtained color $m$ and that edges and non-edges of $G$
can be distinguished by this coloring, that is, edges and non-edges
never obtain the same color. 
For an example of such an edge-colored graph $G$, 
see Figure \ref{fig:symbRepr}. 
The following theorem gives necessary and sufficient conditions on the structure of
graphs $G$ for which one
can find such a coloring and, in addition, provides a new characterization of cographs.

\begin{figure}[tbp]
  \centering
  \includegraphics[bb=  82 567 472 699, scale=0.85]{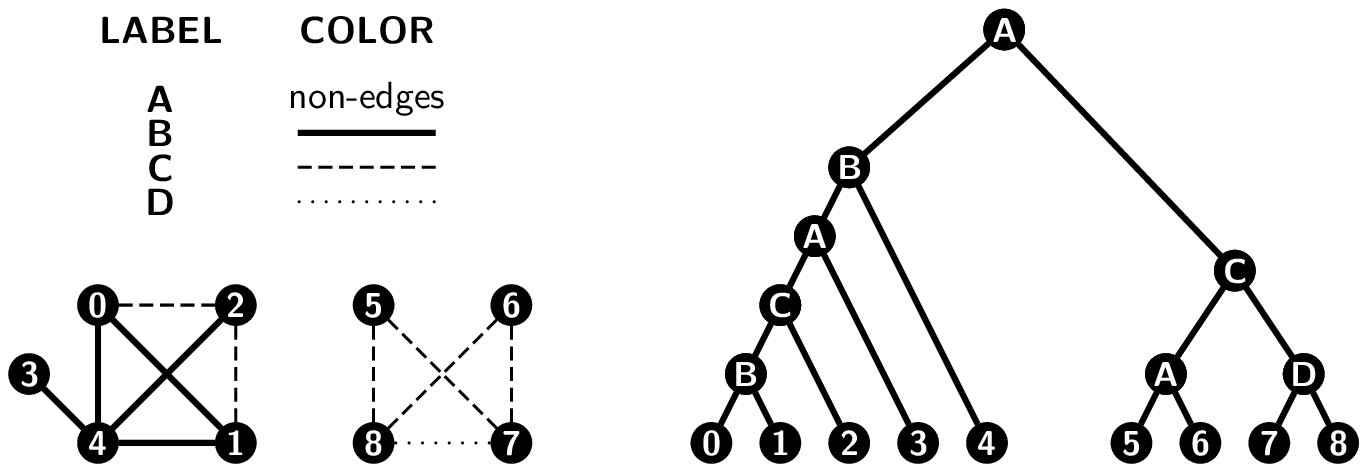} %\hfill
  \caption{Shown is a disconnected edge-colored graph $G$ in the lower left part. 
			  The edge-colors are identified with the labels $B$, $C$  and $D$, as indicated in	
				the upper left part. Non-edges are identified
				with the label $A$. It is easy to verify that the event-labeled tree 
				on the right hand side $(T,t)$  is a tree representation of $G$, 
				since for all distinct leaves $i$ and $j$ we have  $\lca_{T}(i,j) = X\in \{A,B,C,D\}$
				if and only if the (non-)edge $[i,j]$ has the color identified with the respective label $X$. \newline
				In particular, $G$ is a cograph and its cotree representation can be obtained by replacing the label $A$ by $0$, 
					all other labels $B,C,D$ by $1$ and additional contraction of the interior edges
					$[C,B]$ and $[C,D]$. 	  }
  \label{fig:symbRepr}
\end{figure}

\begin{theorem}%[(Non)Edge Distinguishing Symbolic Ultrametrics]
	Let $G=(V,E)$ be an arbitrary 
	(possibly disconnected) graph, $W= \{(x,y)\in V\times V \mid [x,y]\in E\}$
	and $W^c= \{(x,y)\in V\times V \mid [x,y]\not\in E\}$.

	There is a symbolic ultrametric $\delta:V\times V\to \Mo$
	s.t.\ $\de(W)\cap \de(W^c)=\emptyset$ 
%	Then  $\de$ is a symbolic ultrametric 
	if and only if 
	$G$ is a cograph.
	\label{thm:symb-cograph}
\end{theorem}
\begin{proof}
	First assume that $G$ is a cograph. Set $\de(x,x)  = \varnothing$ for 
	all $x\in V$ and set $\de(x,y)=\de(y,x)=1$ if  
	$[x,y]\in E$ and, otherwise, to $0$. 
	Hence, condition $(U0)$ and $(U1)$ are fulfilled. 
	Moreover, by construction $|M|=2$ and thus, Condition $(U2')$ is trivially
	fulfilled. Furthermore, since $G_1(\de)$ and its complement $G_0(\de)$
	are cographs, $(U3')$ is satisfied. Theorem \ref{thm:cograph} implies that
	$\de$ is a symbolic ultrametric. 
	
	Now, let $\delta:V\times V\to \Mo$ be a symbolic ultrametric with 
	$\de(W)\cap \de(W^c)=\emptyset$. 
	Assume for contradiction that $G$ is not a cograph. Then $G$ contains 
	an induced path $P_4 = a-b-c-d$. Therefore, at least one edge $e$
	of this $P_4$
	must obtain a color $\de(e)$ different from the other two edges
	contained in this $P_4$, 
	as otherwise $G_{\de(e)}(\de)$ is not a cograph and 
	thus, $\de$ is not a symbolic ultrametric (Thm. \ref{thm:cograph}, $(U3')$ ).
	For all such possible maps $\de$ ``subdividing'' this 
	$P_4$ we always obtain that two edges of at least one of 
	the underlying paths $P_3 = a-b-c$ or $b-c-d$ 
	must have different colors. W.l.o.g. assume that
	$\de(a,b)\neq \de(b,c)$. Since $[a,c]\not\in E$
	and $\de(W)\cap \de(W^c)=\emptyset$ we can conclude that
	$\de(a,c)\neq \de(a,b)$ and 	$\de(a,c)\neq \de(b,c)$. 
	But then Condition $(U2')$ cannot be satisfied, and 
    Theorem \ref{thm:cograph} implies that $\de$ is not a
	symbolic ultrametric. \hfill \qed
\end{proof}

Theorem \ref{thm:symb-cograph} implies, that there is no hope for finding an
edge-distinguishing map $\delta$
for a graph $G$, that assigns symbols or colors to edges, resp., non-edges such that
for $\delta$ (and hence, for $G$) there is a symbolic representation $(T,t)$, unless $G$ is
already a cograph. 
However, this result does not come as a big surprise, as a
cograph $G$ is characterized by the existence of a unique (up to isomorphism)
cotree $(T',t')$ representing the topology of $G$. 
As a consequence of this result we can infer that 
any symbolic representation $(T,t)$ of a cograph $G$
is a refinement of the cotree representation $(T',t')$ of $G$, %, in symbols $T'\leq T$.
that is, the cotree representation $(T',t')$ of $G$ can be obtained from 
the symbolic representation $(T,t)$ of $\de$ by the following 
procedure:

\bigskip\noindent
First reset for each $v\in V$, 
\begin{equation*}
 t(v) = \begin{cases}
        \varnothing & \textrm{if}   \quad v\in V\setminus V^0 \text{, i.e, } v \text{\ is a leaf} \\
    1 & \textrm{if}\quad  v=\lca_T(x,y) \mbox{ and } \de(x,y)\in W\text{, i.e, } [x,y] \text{\ is an edge in $G$}  \\
    0 & \textrm{if}\quad \mbox{else}\text{, i.e, } [x,y] \text{\ is not an edge in $G$} 
   \end{cases}
\end{equation*}
Clearly, this new map $t$ on the tree $T$ defines a 
symbolic representation $(T,t)$ of the cograph $G=(V,E)$ so that
$[x,y]\in E$ if and only if $t(\lca_T(x,y))=1$.
However, it might be possible that there
is an edge $e=[u,v]\in E^0(T)$ such that $t(u)=t(v)$, and therefore, 
$(T,t)$ is not a cotree representation.
In this case, identify a new vertex $v_e$ with $e$ and define  
the tree $T_e= (V_e,E_e)$ with
vertex set $V_e = V(T) \setminus \{u,v\} \cup \{v_e\}$, edge set $E_e = E(T)
\setminus \{e\} \cup \{[v_e,w] \,:\, [w,u] \mbox{ or } [w,v] \in E \}$, 
that is again a rooted tree. 
Define for all $w\in V_e$ the map
\begin{equation}
t_e(w) = t(w) \textrm{ if } w \neq v_e  \textrm{ and } t(v_e) = t(u).
\end{equation} 
	This construction can be repeated, with $(T_e,t_e)$ now playing the role of $(T,t)$, until a
	we end in a rooted tree $\hat T=(\hat V,\hat E)$ with a map $\hat t:\hat V\to \Mo$ 
	so that for all edges
	$[u,v]\in  \hat E^0$ it holds that $\hat t(u)\neq \hat t(v)$.

\bigskip
With this procedure, we obtain a symbolic representation $(\hat T, \hat t)$
of the cograph $G$, also known as so-called \emph{discriminating} 
symbolic ultrametric \cite{HHH+13}.
In particular, this representation $(\hat T, \hat t)$
is unique (up to isomorphism) \cite[cf.\ Prop.\ 1]{HHH+13} and,
by construction, satisfies the condition of 
a cotree representation. Moreover, since the cotree representation $(T',t')$ is unique 
(up to isomorphism) \cite{lerchs71,lerchs72}, it follows that
that $(T',t')$ and $(\hat T, \hat t)$
must be isomorphic. We summarize this result in the following corollary.

\begin{corollary}
	Let $G=(V,E)$ be a cograph, $(T',t')$ be the corresponding cotree representation,	
	and $W$, resp., $W^c$ as defined in Theorem \ref{thm:symb-cograph}. 
%	$W= \{(x,y)\in V\times V \mid [x,y]\in E\}$
%	and $W^c= \{(x,y)\in V\times V \mid [x,y]\not\in E\}$.
	Moreover, assume that there is a symbolic ultrametric $\delta:V\times V\to \Mo$
	s.t.\ $\de(W)\cap \de(W^c)=\emptyset$ with $(T,t)$ being the corresponding
	symbolic representation of $\de$.

	Assume that the pair  $(\hat T, \hat t)$ is obtained from $(T,t)$ by application
	of the procedure above. 
	Then,  $(\hat T, \hat t)$ and $(T',t')$ are isomorphic.
\end{corollary}

Assume that we want to find a symbolic ultrametric that can distinguish
between ``most of'' the edges and/or non-edges, however, the given graph is a non-cograph $G=(V,E)$.
Then, we are immediately left with the following problems. 

\begin{problem1}\textsc{Cograph Editing/Deletion/Completion}

\vspace{0.3em}
\begin{tabular}{ll}
	\emph{Input:}& Given a simple graph $G=(V,E)$ and an integer $k$.\\
	\emph{Question:}& Is there a cograph $G'=(V,E')$, s.t.\ \\
	                & $E'\subseteq {V\choose 2}$ and $|E \Delta E'| \leq k$ (Editing), \\
						 & $E'\subseteq E$ and $|E \setminus E'| \leq k$  (Deletion), or \\
						 & $E\subseteq E'$ and $|E' \setminus E| \leq k$  (Completion).
\end{tabular}
\end{problem1}

% Clearly, for $G'$ one can easily derive such a symbolic ultrametric $\delta$ by
% setting $\delta(x,y)=1$ for all $[x,y]\in E(G')$ and $\delta(x,y)=0$ for all
% $[x,y]\notin E(G')$ resulting in the cotree $(T,t)$ for $G^*$. 
However, the (decision version of the) problem to edit a given graph $G$ into a
cograph $G'$, and thus, to find the closest graph $G'$ that has a symbolic
representation, is NP-complete \cite{Liu:11,Liu:12}.
In addition, the problems of deciding whether there is a cograph $G'$
resulting by adding, resp., removing $k$ edges from $G$ is NP-complete, as well
\cite{EMC:88}.

\begin{theorem}[Liu et al.\ \cite{Liu:12}, El-Mallah and Colbourn \cite{EMC:88}]
\textsc{Cograph Editing}, \textsc{Cograph Completion} 
and \textsc{Cograph Deletion} are NP-complete. 
\end{theorem}

In what follows, we will consider and discuss two modifications of the problem of finding
a symbolic ultrametric that can distinguish between edges and non-edges in Section
\ref{sec:sue} and \ref{sec:crd}: 
\begin{enumerate}
\item In Section \ref{sec:sue} we consider a couple of problems which are 
		of highly practical relevance:
      The symbolic ultrametric editing, completion and deletion problem. 
%		For an arbitrary given map $d:V\times V\to \Mo$ we ask for the minimum
%		number of changing the values $d(x,y) = m\in \Mo$ for elements $x,y\in V$
%		so that the resulting map $\de$ is a symbolic ultrametric. This problem is equivalent
%		to ask 
%		for the minimum number of color-changes of the edges 
%		of an arbitrary edge-colored complete graph (which is a cograph) 
%		to obtain 
%		a symbolic ultrametric, and hence, a symbolic representation of the graph.
%		Additionally, we address the problem of adding or deleting 
%		a minimum number of edges of arbitrary edge-colored graph $G=(V,E)$
%		\\
\item In contrast, if a non-edge colored graph $G$ is not a cograph and thus, if 
	   there is no single tree representation of $G$, then 
	   we ask for the minimum number of trees that are needed in order
	   to represent the topology of $G$ in an unambiguous way, 
		see Section \ref{sec:crd}. 
\end{enumerate}

\section{Symbolic Ultrametric Editing, Completion and Deletion}
\label{sec:sue}

Symbolic ultrametrics lie at the heart of many problems in phylogenomics.
Phylogenetic Reconstructions are concerned with the study of the evolutionary
history of groups of systematic biological units, e.g. genes or species. The
objective is the assembling of so-called phylogenetic trees or networks that
represent a hypothesis about the evolutionary ancestry of a set of genes,
species or other taxa.

Genes are passed from generation to generation to the offspring. Some of those
genes are frequently duplicated, mutate or get lost - a mechanism that also
ensures that new species can evolve. 
Crucial for the evolutionary reconstruction of species history is the knowledge
of the relationship between the respective genes. Genes that share
a common origin (homologs) are divided into three classes, namely orthologs,
paralogs, and xenologs \cite{Fitch2000}. Two homologous genes are orthologous if
at their most recent point of origin the ancestral gene complement is
transmitted to two daughter lineages; a speciation event happened. They are
paralogous if the ancestor gene at their most recent point of origin was
duplicated within a single ancestral genome; a duplication event happened.
Horizontal gene transfer (HGT) refers to the transfer of genes between organisms
in a manner other than traditional reproduction and across different species; if
such an event happened at the most recent point of origin of two genes, then
they are called xenologous.
Intriguingly, there are practical sequence-based methods that allow to
determine whether two genes $x$ and $y$ are orthologs or not
with acceptable accuracy \emph{without} constructing either gene or species trees 
\cite{Lechner:11a,Lechner:14}.

Now, assume we have given an estimate of genes being orthologs, 
paralogs or even xenologs, that is a map $d:X\times X\to \{\text{speciation,\ duplication,\ HGT}\}$. 
Then, one is interested in the representation of these estimates as
a tree $T$ with event-labeling $t$ so-that $t(\lca(x,y))=\text{speciation}$
iff $x,y$ are orthologs, $t(\lca(x,y))=\text{duplication}$ iff $x,y$ are paralogs and
$t(\lca(x,y))=\text{HGT}$ iff $x,y$ are xenologs. 
In practice, however, such maps $d$ are often only estimates of the true evolutionary 
relationship $\de$ between the investigated genes. 
Thus, in general such estimates $d$ will not be a symbolic ultrametric. 
Hence, there is a big interest in optimally editing $d$ to a
symbolic ultrametric $\de$.

The problem of editing a given symmetric map $d : X \times X \to \Mo$ to a
symbolic ultrametric is defined as follows:

\begin{problem1}\textsc{Symbolic Ultrametric Editing}

\vspace{0.3em}
\begin{tabular}{ll}
	\emph{Input:}& Given a symmetric map $d : X \times X \to \Mo$, s.t.\ \\
					&  $d(x,y)=\varnothing$ if and only if $x=y$.\\
	\emph{Question:}& Is there a symbolic ultrametric $\delta : X \times X \to \Mo$, s.t.\ for \\
						& $D = \{(x,y)\in X\times X \mid d(x,y) \neq \delta(x,y)\}$ we have $|D| \leq k$.
\end{tabular}
\end{problem1}

A further problem arising from the latter considerations is as follows. 
Assume we have an assignment of a symmetric subset $R$ of $X\times X$ so that for
all $(x,y)\in R$ the assignment $d(x,y)$ is believed to be an reliable estimate 
and thus, which is not allowed to be changed. 
Moreover, let $X\times X\setminus R$
be the pairs $(x,y)$ for which an assignment $d(x,y)$ is not known.
%i.e., they might correspond to any assignment $m\in \Mo$. 
Assume that $M=\{1,\dots,n\}$ and $\Mo = M\cup\{\varnothing,0\}$, then we can 
extend the map $d\colon X\times X\to \Mo$ so that
\begin{equation*}
 d(x,y) = \begin{cases}
       	 \varnothing & \textrm{if} \quad  x=y \\
    d(x,y) & \textrm{if}\quad  (x,y)\in R  \\
    0 & \textrm{if}\quad (x,y)\in X\times X\setminus R 
   \end{cases}
\end{equation*}
We then ask to change the assignment of a minimum number of pairs $(x,y)$ 
with $d(x,y)=0$ to some element in $m\in M, m\neq 0$ so that the resulting map is a symbolic ultrametric. 
In other words, only non-reliable estimates of pairs $(x,y)$ are allowed to be changed.

\begin{problem1}\textsc{Symbolic Ultrametric Completion}

\vspace{0.3em}
\begin{tabular}{ll}
	\emph{Input:}& Given a symmetric map $d : X \times X \to \Mo$, s.t.\ \\
					&  $d(x,y)=\varnothing$ if and only if $x=y$.\\
	\emph{Question:}& Is there a symbolic ultrametric $\delta : X \times X \to \Mo$ s.t.\ \\
						 &  if $d(x,y)\neq 0$, then $\delta(x,y)=d(x,y)$; and  $|D| \leq k$, where  \\
						 & $D = \{(x,y)\in X\times X \mid d(x,y) \neq \delta(x,y)\}$.
\end{tabular}
\end{problem1}

Conversely, one might ask 
to change a minimum number of assignments $d(x,y)\neq 0$ to $\de(x,y)=0$.

\begin{problem1}\textsc{Symbolic Ultrametric Deletion}

\vspace{0.3em}
\begin{tabular}{ll}
	\emph{Input:}& Given a symmetric map $d : X \times X \to \Mo$, s.t.\ \\
					&  $d(x,y)=\varnothing$ if and only if $x=y$.\\
	\emph{Question:}& Is there a symbolic ultrametric $\delta : X \times X \to \Mo$ s.t.\ \\
						 & $\delta(x,y)=d(x,y)$ or $\delta(x,y) = 0$; and  $|D| \leq k$, where  \\
						& $D = \{(x,y)\in X\times X \mid d(x,y) \neq \delta(x,y)\}$.
\end{tabular}
\end{problem1}

%For a complete prob
%On the other hand, of one trusts in all assignments with some fixed value,
%say $d(x,y) = 0$, but not in the others, then one
%one might ask for the change of the assignment of a minimum number of 
%pairs $(x,y)$ with $d(x,y)\neq 0$ to some element in $m\in M, m\neq0$ 
%so that the resulting map is a symbolic ultrametric. 

%\begin{problem1}\textsc{Symbolic Ultrametric Full Completion}

%\vspace{0.3em}
%\begin{tabular}{ll}
%	\emph{Input:}& Given a symmetric map $d : X \times X \to \Mo$ with \\ 
%					 &	$M=\{0,\ldots,n\}$, s.t.\ %\\
%					  $d(x,y)=\varnothing$ if and only if $x=y$.\\
%	\emph{Question:}& Is there a symbolic ultrametric $\delta : X \times X \to \Mo\setminus \{0\}$, s.t.\  \\
% 			 			 & $\delta(x,y)=d(x,y)$ if $d(x,y)\neq 0$ and $\delta(x,y)\in \Mo$ if $d(x,y) = 0$
%						
%\end{tabular}
%\end{problem1}

\subsection{Computational Complexity}

In this section, 
we prove the NP-completeness of \textsc{Symbolic Ultrametric Editing},  
\textsc{Symbolic Ultrametric Completion} and \textsc{Symbolic Ultrametric Deletion}.

\begin{theorem}
\textsc{Symbolic Ultrametric Editing} is NP-complete. 
\label{thm:CE-NPC}
\end{theorem}

\begin{proof}
Given a symmetric map $\delta$ it can be verified in polynomial time, if
$\delta$ is a symbolic ultrametric: One can check Conditions (U2) and
(U3) individually for each of the $O(|X|^3)$ many combinations of $\{x,y,z\} \in {X
\choose 3}$ for (U2), and the $O(|X|^4)$ many combinations of $\{x,y,u,v\} \in {X
\choose 4}$ for (U3), respectively. Hence, $\textsc{Symbolic Ultrametric
Editing} \in NP$. We will show by reduction from \textsc{Cograph Editing} that
\textsc{Symbolic Ultrametric Editing} is NP-hard.

Let $G=(V,E)$ be an arbitrary simple graph. 
We associate with $G$ a map $d \colon V \times V \to\Mo$, where 
$M = \{0,1,\ldots,n\}$ is a non-empty
finite set s.t.\ $n\geq 1$ and thus, $0,1 \in M$. 
Let $\Mo := M \cup \{\varnothing\}$ and 
set for all $x,y\in V$:
\begin{equation*}
 d(x,y) = d(y,x) = \begin{cases}
       	 \varnothing & \textrm{if} \quad  x=y \\
            1 & \textrm{if}\quad  [x,y]\in E  \\
    			0 & \textrm{if}\quad [x,y]\notin E
   \end{cases}
\end{equation*}
Obviously, $d$ can be constructed in polynomial time.
In the following, we show, that given an integer $k$, there exists a solution of
the \textsc{Cograph Editing} problem for $G$ and integer $k$ if and only if there exists
a solution of the \textsc{Symbolic Ultrametric Editing} problem for $d$ and
integer $2k$.

First, we show that a solution of the \textsc{Symbolic Ultrametric Editing}
problem for $d$ and $2k$ can be constructed from a solution of the
\textsc{Cograph Editing} problem for $G$ and $k$. Let $G'=(V,E')$ be a cograph
with $|E \Delta E'| \leq k$. Furthermore let $\delta \colon V \times V \to \Mo$ be a
map, such that for all $x,y \in V$, 
\begin{equation*}
 \de(x,y) = \de(y,x) = \begin{cases}
	       	 \varnothing & \textrm{if} \quad  x=y \\
            1 & \textrm{if}\quad  [x,y]\in E'  \\
    			0 & \textrm{if}\quad [x,y]\notin E'
   \end{cases}
\end{equation*}
%$\delta(x,y)=\delta(y,x)=1$ if and
%only if $[x,y] \in E'$ and $\delta(x,y)=\delta(y,x)=0$ if and only if $[x,y]
%\notin E'$. Additionally, set $\delta(x,x)=\varnothing$ for all $x \in V$. 
It is easy to verify that $\delta$ is a symbolic ultrametric by application of Theorem
\ref{thm:cograph}. 
It remains to show that for $D=\{(x,y)\in X\times X \mid d(x,y) \neq \delta(x,y)\}$ it holds
that $|D| \leq 2k$. Note that for all $x \in V$ we have 
$d(x,x)=\delta(x,x)=\varnothing$ and therefore $(x,x) \notin D$.
The set $D$ can be partitioned into the two subsets
\begin{align*}
        D_1&=\{(x,y) \mid d(x,y)=1 \wedge \delta(x,y)=0\} \text{ and }\\
        D_2&=\{(x,y) \mid d(x,y)=0 \wedge \delta(x,y)=1\}.
\end{align*}
Hence, $(x,y) \in D_1$ if and only if $[x,y] \in E \setminus E'$, 
and $(x,y) \in D_2$ if and only if $[x,y] \in E' \setminus E$.
As $(E \setminus E') \cup (E' \setminus E) = (E \Delta E')$
it holds that, $(x,y) \in D$ if and only if $[x,y] \in E \Delta E'$.
As $d$ and $\delta$ are symmetric, it also holds that $(x,y) \in D$ if and 
only if $(y,x) \in D$.
Hence, $[x,y] \in E \Delta E'$ if and only if $(x,y) \in D$ and $(y,x) \in D$.
This reflects the fact, that an edge edit $[x,y] \in E \Delta E'$ in $G$
corresponds to the two symmetric edits $(x,y), (y,x) \in D$ in $d$.
Therefore, $|D| = |\{(x,y) \mid d(x,y) \neq \delta(x,y)\}| = 2|E \Delta E'| \leq 2k$.

We continue to show that a solution of the \textsc{Cograph Editing} problem for
$G$ and $k$ can be constructed from a solution of the \textsc{Symbolic
Ultrametric Editing} problem for $d$ and $2k$. Let $\delta : V \times V \to \Mop$
be a symbolic ultrametric s.t.\ $|D| = |\{(x,y) \mid d(x,y) \neq \delta(x,y)\}|
\leq 2k$. 
%Note, the latter allows both $\Mop\subseteq \Mo$ or $\Mo\subseteq \Mop$,
%however, in all cases we assume that $1\in \Mop$. 
Furthermore, let $G'=(V,E')$ be a simple graph, such that for all $x,y
\in V$ it holds that $[x,y] \in E'$ if and only if $\delta(x,y)=1$.
By Theorem \ref{thm:cograph} (U3') we have that $G'=G_1$ and hence, $G'$ is a
cograph. It remains to show that $|E\Delta E'|\leq k$. By construction, 
for all $x \in V$, $d(x,x)=\delta(x,x)=\varnothing$ and $[x,x] \notin E \Delta E'$.
Let $D = \{(x,y) \mid d(x,y) \neq \delta(x,y)\}$.
Note that for all distinct $x,y \in V$ it holds that $d(x,y) \in \{0,1\}$.
Hence, $D$ can be partitioned into the four subsets
\begin{align*}
D_1 &= \{(x,y) \mid d(x,y)=1 \wedge \delta(x,y)=0\},\\
D_2 &= \{(x,y) \mid d(x,y)=0 \wedge \delta(x,y)=1\},\\
D_3 &= \{(x,y) \mid d(x,y)=1 \wedge \delta(x,y)\in \Mop\setminus\{0,1\}\}, \text{ and }\\ %\{3,\ldots,n\}\}, \text{ and }\\
D_4 &= \{(x,y) \mid d(x,y)=0 \wedge \delta(x,y)\in \Mop\setminus\{0,1\}\}.
\end{align*}
For these subsets of $D$ it holds that
if $(x,y) \in D_1$ then $[x,y] \in E \setminus E'$, and
if $(x,y) \in D_2$ then $[x,y] \in E' \setminus E$.
Furthermore, $\delta(x,y)\in \Mop\setminus\{0,1\}$ %\{3,\ldots,n\}$ 
implies that $[x,y] \notin E'$ and it follows that
if $(x,y) \in D_3$ then $[x,y] \in E \setminus E'$, and
if $(x,y) \in D_4$ then $[x,y] \notin E \wedge [x,y] \notin E'$.
For all remaining $x,y \in V$, i.e., for which $d(x,y)=\delta(x,y)$, 
it holds that %$(x,y) \in E \wedge (x,y) \in E'$.
$[x,y] \notin E \setminus E'$ and $[x,y] \notin E' \setminus E$.
It follows that $[x,y] \in E \setminus E'$ if and only if $(x,y) \in D_1 \cup D_3$, and
$[x,y] \in E' \setminus E$ if and only if $(x,y) \in D_2$.
As before, due to the symmetry of the maps $d$ and $\delta$, two symmetric edits $(x,y), (y,x) \in D$ in $d$
correspond to at most one edge edit $[x,y] \in E \Delta E'$ in $G$.
Finally, $2|E \Delta E'| = 2|E \setminus E'| + 2|E' \setminus E|) = |D_1 \cup D_3| + |D_2| \leq |D| \leq 2k$.
Hence, $|E \Delta E'| \leq k$.

Thus, \textsc{Symbolic Ultrametric Editing} is NP-complete.\hfill \qed
\end{proof}

\begin{theorem}
\textsc{Symbolic Ultrametric Completion} is NP-complete. 
\label{thm:SUC-NPC}
\end{theorem}

\begin{proof}
It is shown analogously as in the proof of Theorem \ref{thm:CE-NPC}
that $\textsc{Symbolic Ultrametric Min Completion}\in NP$. 
We will show by reduction from \textsc{Cograph Completion} that
\textsc{Symbolic Ultrametric Completion} is NP-hard.

Let $G=(V,E)$ be an arbitrary simple graph. We associate to 
$G$ a map $d \colon V \times V \to
\Mo$ as defined in the proof of Theorem \ref{thm:CE-NPC}:
\begin{equation*}
 d(x,y) = d(y,x) = \begin{cases}
       	 \varnothing & \textrm{if} \quad  x=y \\
            1 & \textrm{if}\quad  [x,y]\in E  \\
    			0 & \textrm{if}\quad [x,y]\notin E
   \end{cases}
\end{equation*}

Let there be a solution $G'=(V,E')$ for the  \textsc{Cograph Completion} problem 
for $G$ and $k$, i.e., $E\subseteq E'$ and $|E'\setminus E|\leq k$. 
We show that that there is a solution for the \textsc{Symbolic Ultrametric Completion}
problem for $d$ and $2k$. Define the map $\delta \colon V \times V \to \Mo$ 
as in the proof of Theorem \ref{thm:CE-NPC}:
\begin{equation*}
 \de(x,y) = \de(y,x) = \begin{cases}
	       	\varnothing & \textrm{if} \quad  x=y \\
            1 & \textrm{if}\quad [x,y]\in E'     \\
    			0 & \textrm{if}\quad [x,y]\notin E'
   \end{cases}
\end{equation*}
Again, it is easy to verify that $\delta$ is a symbolic ultrametric by application of Theorem
\ref{thm:cograph}. Moreover, by construction $\de(x,y)=d(x,y)$ for all $x,y\in V$ whenever
$[x,y]\in E\subseteq E'$ and hence, for all $x,y\in V$ with $d(x,y)\neq 0$.

It remains to show that for $D=\{(x,y)\in X\times X \mid 0=d(x,y) \neq \delta(x,y)\}$ it holds
that $|D| \leq 2k$. 
Note that for all $x \in V$ we have 
$d(x,x)=\delta(x,x)=\varnothing$ and therefore $(x,x) \notin D$.
Moreover, \[D=\{(x,y) \mid d(x,y)=0 \wedge \delta(x,y)=1\}.\]
Hence,  $(x,y),(y,x) \in D$ if and only if $[x,y] \in E' \setminus E$.
Therefore, $|D|  = 2| E'| \leq 2k$.

We continue to show that a solution of the \textsc{Cograph Editing} problem for
$G$ and $k$ can be constructed from a solution of the \textsc{Symbolic
Ultrametric Editing} problem for $d$ and $2k$. Let $\delta : V \times V \to \Mop$
be a symbolic ultrametric s.t.\ $|D| \leq 2k$ and $\de(x,y)=d(x,y)$
if $d(x,y)\neq 0$. 
%Note, the latter allows both $\Mop\subseteq \Mo$ or $\Mo\subseteq \Mop$,
%however, in all cases we assume that $1\in \Mop$. 
Furthermore, let $G'=(V,E')$ be a simple graph, such that for all 
$x,y\in V$ it holds that $[x,y] \in E'$ if and only if $\delta(x,y)=1$.
By Theorem \ref{thm:cograph} (U3') we have that $G'=G_1$ and hence, $G'$ is a
cograph. 
It remains to show that $|E'\setminus E|\leq k$. By construction, 
for all $x \in V$, $d(x,x)=\delta(x,x)=\varnothing$ and $[x,x] \notin E'$.
Note that for all distinct $x,y \in V$ it holds for the map associated
to $G$ that $d(x,y) \in \{0,1\}$.
Hence, $D$ can be partitioned into 
\begin{align*}
D_1 &= \{(x,y) \mid d(x,y)=0 \wedge \delta(x,y)=1\}, \text{ and }\\
D_2 &= \{(x,y) \mid d(x,y)=0 \wedge \delta(x,y)\in \Mop\setminus\{0,1\}\}.
\end{align*}
Thus,  if $(x,y),(y,x) \in D_1$, then $[x,y] \in E' \setminus E$. 
Therefore, $2(|E' \setminus E|)  = |D_1| \leq |D| \leq 2k$
and thus, $|E' \setminus E| \leq k$.

Hence, \textsc{Symbolic Ultrametric Completion} is NP-complete.\hfill \qed
\end{proof}

Using similar arguments as in the proof of Theorem \ref{thm:SUC-NPC}
we can infer the NP-completeness of \textsc{Symbolic Ultrametric Deletion}
by reduction from \textsc{Cograph Deletion}.

\begin{theorem}
\textsc{Symbolic Ultrametric Deletion} is NP-complete. 
\label{thm:SUD-NPC}
\end{theorem}

\subsection{Integer Linear Program}

We showed in \cite{HLS+14} that the cograph editing problem is amenable to
formulations as Integer Linear Program (ILP). We will extend these 
results here to solve the symbolic ultrametric editing/completion/deletion problem. 
Let $d:X\times X\to \Mo$ be an arbitrary symmetric map with $M=\{0,1,\ldots,n\}$
and $K_{|X|}=(X,E={X \choose 2})$ be the corresponding complete
graph with edge-coloring s.t.\ each edge $[x,y]\in E$ obtains color $d(x,y)=d(y,x)$.

For each of the three problems and hence, a given 
 symmetric map $d$ we define for each distinct $x,y\in X$ and $i\in M$
the binary constants $\mathfrak{d}^i_{x,y}$ with $\mathfrak{d}^i_{x,y}=1$ if and only if $d(x,y)=i$.
Moreover, we define the binary variables $E^i_{xy}$ for all $i \in M$ and $x,y \in X$
that reflect the coloring of the edges in $K_{|V|}$ of the final symbolic
ultrametric $\de$, i.e.,  $E^i_{xy}$ is set to $1$ if and only if 
$\de(x,y)=i$. 

%\paragraph{Symbolic Ultrametric Editing}

In order to find the closest symbolic ultrametric $\de$, 
the objective function is to minimize the symmetric difference of the 
$d$ and $\de$ among all different symbols $i\in M$:
\begin{equation}
\min \sum_{i \in M} \Bigg( \sum_{(x,y) \in X} (1-\mathfrak{d}^i_{xy})E^i_{xy} + \sum_{(x,y) \in X} \mathfrak{d}^i_{xy}(1-E^i_{xy}) \Bigg)
\end{equation}
The same objective function can be used for
the symbolic ultrametric completion and deletion problem. 

\smallskip
In case of the the symbolic ultrametric completion
we must ensure that $\de(x,y)=d(x,y)$ for all $d(x,y)\neq0$.
Hence we set for all $x,y$ with $d(x,y)=i\neq 0$:
\begin{align}
E^i_{x,y} =1. 
 \label{eee1}
\end{align}

In case of  the symbolic ultrametric deletion
we must ensure that $\de(x,y)=d(x,y)$ or $\de(x,y)=0$
or, in other words, for all $d(x,y)=i\neq 0$ it must
hold that either $E^i_{xy}=1$ or $E^0_{xy}=1$
Hence, we set for for all for all $x,y\in V$:
\begin{align}
	E^0_{xy}=1, \mbox{ if } d(x,y) = 0 \text{, and }  E^i_{xy}+E^0_{xy}=1,  \text{ else.} \tag{\ref{eee1}'}
   %\mathfrak{d}^i_{x,y} - E^i_{x,y} = 0 \text{ for all } i\in\{1,\ldots,n\} 
\label{eee2}
\end{align}

For the cograph editing problem we neither need Constraint
\ref{eee1} nor \ref{eee2}. 
However, for all three problems we need the following.

\smallskip
Each tuple $(x,y)$ with $x \neq y$ has exactly one value $i \in M$ assigned to it
which is expressed in the following constraint.
\begin{equation}
\sum_{i \in M} E^i_{x,y}=1 \text{ and }
E^i_{xy}-E^i_{yx}=0 \text{ for all } x,y \in X.
\end{equation}

In order to satisfy Condition (U2') and thus, that all induced triangles have at
most two colors on the edges we need this constraint.
\begin{equation}
E^i_{xy}+E^j_{yz}+E^k_{xz} \leq 2 \label{ilp:u2}
\end{equation}
for all ordered tuples $(i,j,k)$ of distinct $i,j,k \in M$ and pairwise distinct $x,y,z \in X$.

Finally, in order to satisfy Condition (U3')  and thus, that each mono-chromatic subgraph
comprising all edges with fixed color $i$ is a cograph, we need the following constraint
that forbids induced $P_4$'s.
\begin{equation}
E^i_{xy} + E^i_{yu} + E^i_{uv} - E^i_{xu} - E^i_{xv} - E^i_{yv} \leq 2
\end{equation}
for all $i \in M$ and all ordered tuples $(x,y,u,v)$ of distinct $x,y,u,v \in X$.

It is easy to verify that the latter ILP formulation needs 
$O(|M||X|^2)$ variables and $O(|M|^3|X|^3+|X|^4)$ constraints.

\section{Cotree Representation and Cograph $\boldsymbol{k}$-Decomposition}
\label{sec:crd}

If a given non-edge colored graph $G$ is not a cograph, then 
Theorem \ref{thm:symb-cograph} implies that one cannot define
an edge-distinguishing symbolic ultrametric, and thus, 
in particular no single tree representation of $G$.
Therefore, we are interested to represent the topology of $G$ in an unambiguous way
with a minimum number of trees.

Recollect, a graph $G = (V,E)$ is represented by a set of cotrees $\mathbb T = \{T_1,\dots,T_k\}$, 
if and only if for each edge $[x,y]\in E$ there is a tree $T_i\in \mathbb T$
with $t(\lca_{T_i}(x,y)) =1$. 

Note, by definition, each cotree $T_i$ determines a subset 
$E_i = \{[x,y]\in E \mid t(\lca_{T_i}(x,y))=1\}$ of $E$. Hence, the subgraph 
$(V,E_i)$ of $G$ must be a cograph. Therefore, in order to find the minimum number of cotrees
representing a graph $G$, we can equivalently ask for a decomposition 
$\Pi=\{E_1,\dots,E_k\}$ of $E$ so that each subgraph $(V,E_i)$ is a cograph, 
where $k$ is the least integer among all cograph decompositions of $G$. 
Thus, we are dealing with the following two equivalent problems.

\begin{problem1}\textsc{Cotree $k$-Representation}

\vspace{0.3em}
\begin{tabular}{ll}
	\emph{Input:}& Given a graph $G=(V,E)$ and an integer $k$ . \\
	\emph{Question:}& Can $G$ be represented by $k$ cotrees?
\end{tabular}
\end{problem1}

\begin{problem1}\textsc{Cograph $k$-Decomposition}

\vspace{0.3em}	
\begin{tabular}{ll}
  \emph{Input:}& Given a graph $G=(V,E)$ and an integer $k$. \\
\emph{Question:}& Is there a cograph $k$-decomposition of $G$?
\end{tabular}
\end{problem1}

Clearly, any cograph has an optimal $1$-decomposition, while
for cycles of length $>4$ or paths $P_4$ there is always an 
optimal cograph 2-decomposition. 
However, there are examples of graphs that even do not have
a cograph 2-decomposition, see Figure \ref{fig:3col}.
Moreover, as shown in Figure \ref{fig:exp}, 
the number of different 
optimal cograph $k$-decomposition on a graph can grow exponentially.
the next theorem provides  a non-trivial upper bound for the integer $k$ s.t.\ there 
is still a cograph $k$-decomposition for arbitrary graphs.

\begin{theorem}
For every graph $G$ with maximum degree $\Delta$
there is a cograph $k$-decomposition with 
$1\leq k\leq \Delta+1$ that can 
be computed in $O(|V||E| + \Delta(|V|+|E|))$ time.
Hence, any graph can
be represented by at most $\Delta+1$ cotrees.
\label{thm:k-delta}
\end{theorem}
\begin{proof}
Consider a proper edge-coloring $\varphi:E\to \{1,\dots,k\}$ of $G$, 
i.e., an edge coloring such that no two incident edges obtain the same color.
Any proper edge-coloring using $k$ colors 
yields a cograph $k$-partition $\Pi= \{E_1,\dots,E_k\}$ where
$E_i=\{e\in E\mid \varphi(e)=i\}$,
because any connected component in $G_i =(V,E_i)$ 
is an edge and thus, no $P_4$'s are contained in $G_i$.
%since each subgraph 
%$G_i=(V,E_i)$ is a cograph as any connected component in $G_i$ 
Vizing's Theorem \cite{V:64} implies that for each graph 
there is a proper edge-coloring using $k$ colors with $\Delta\leq k\leq \Delta+1$. 
%The resulting partition $\Pi$ of the edge set is a cograph $k$-partition, and thus, 
%a cograph $k$-decomposition.   

An proper edge-coloring using at most $\Delta+1$ colors can be computed 
with the Misra-Gries-algorithm in $O(|V||E|)$ time \cite{MG:92}. 
Since the (at most $\Delta+1$) respective cotrees can be constructed in linear-time $O(|V|+|E|)$ \cite{corneil1985linear}, we
derive the runtime $O(|V||E| + \Delta(|V|+|E|))$. \hfill \qed
\end{proof}

\begin{figure}[tbp]
  \centering
  \includegraphics[bb= 182 515 409 694, scale=0.6]{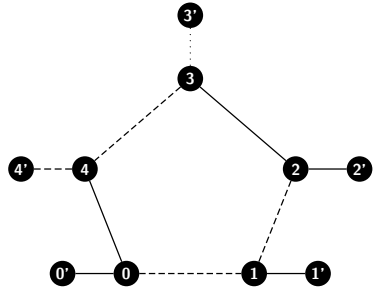} %\hfill
  \caption{Full enumeration of all possibilities (which we leaf to the reader), 
			shows that the depicted graph has no cograph 2-decomposition. 
			The existing 
		    cograph 3-decomposition is also a cograph 3-partition; 
			highlighted by dashed-lined, dotted and bold edges.} 
  \label{fig:3col}
\end{figure}

\begin{figure}[tbp]
  \centering
  \includegraphics[bb=   112 370 478 657, scale=0.5]{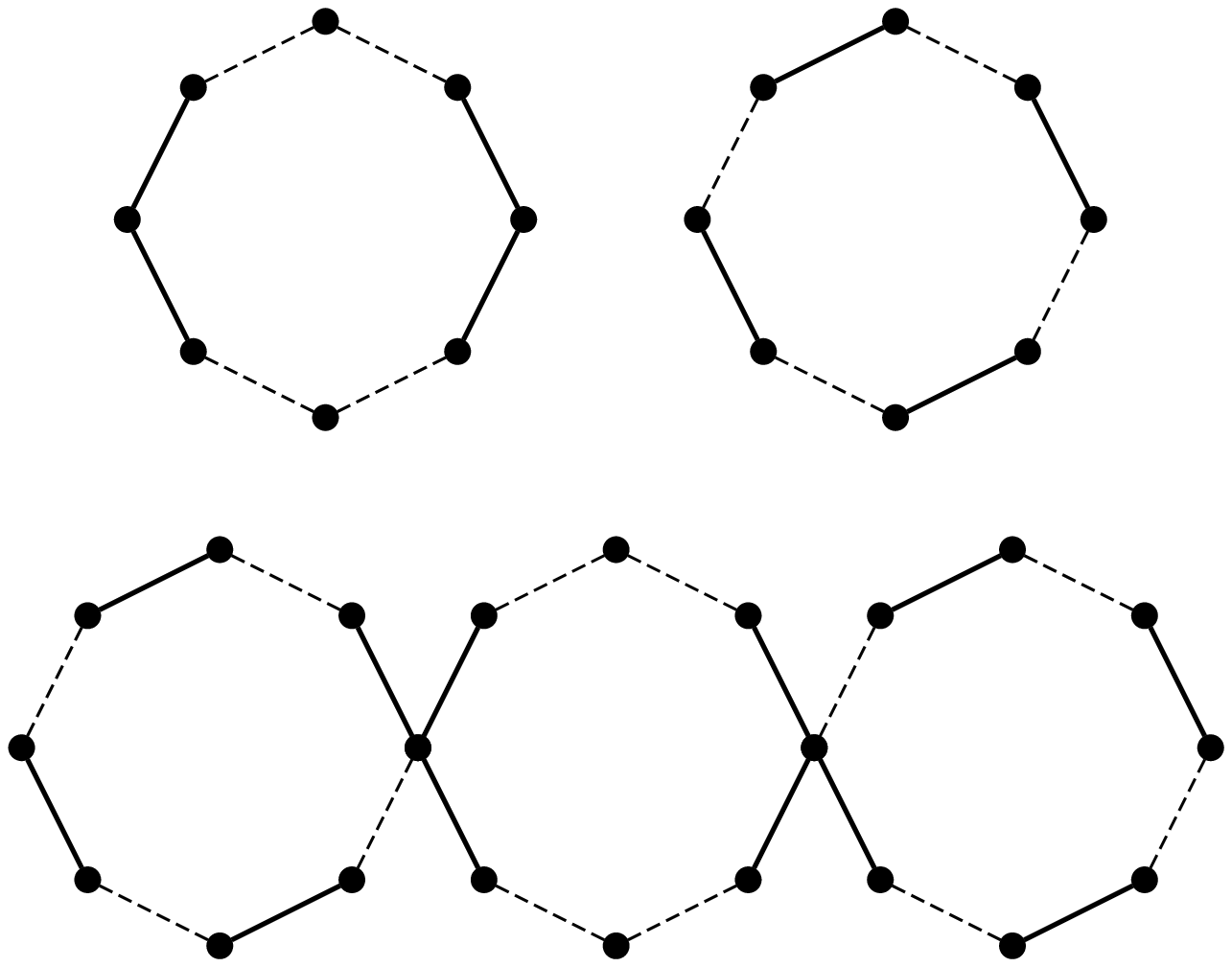} %\hfill
  \caption{Two isomorphic graphs with two non-equivalent optimal cograph 2-decomposition
			  (highlighted by dashed and solid edges) are shown in the upper part. 
			  By stepwisely identifying single vertices  one
           obtains a chain of graphs $G$, see lower part. For each subgraph that is a copy of the
           graph above, an optimal cograph 2-decomposition can be determined almost independently of the remaining
			  parts of the graph $G$. Hence, with an increasing number of vertices of such chains $G$
			  the number of different cograph 2-decompositions is growing exponentially.}
  \label{fig:exp}
\end{figure}

Obviously, any optimal $k$-decomposition must also be a coarsest
$k$-decomposition, while the converse is in general not true, see
Fig.\ \ref{fig:pd}.
The partition $\Pi= \{E_1,\dots,E_k\}$ 
obtained from a proper edge-coloring is usually not a coarsest one, 
as possibly $(V,E_J)$ is a cograph, where
$E_J=\cup_{i\in J} E_i$ and  $J\subseteq \{1,\dots,l\}$.
However, there are graphs having an \emph{optimal} cograph $\Delta$-decomposition, 
see Fig.\ \ref{fig:3col} and \ref{fig:exp}. 
Thus, the 
derived bound $\Delta+1$ is almost sharp.
Nevertheless, we assume that this bound can be sharpened: 

\begin{conjecture}
For every graph $G$ with maximum degree $\Delta$
there is a cograph $\Delta$-decomposition.
	\label{conj:k-partDelta}
\end{conjecture}

\begin{figure}[tbp]
  \centering
  \includegraphics[bb= 6 416 435 651, scale=0.6]{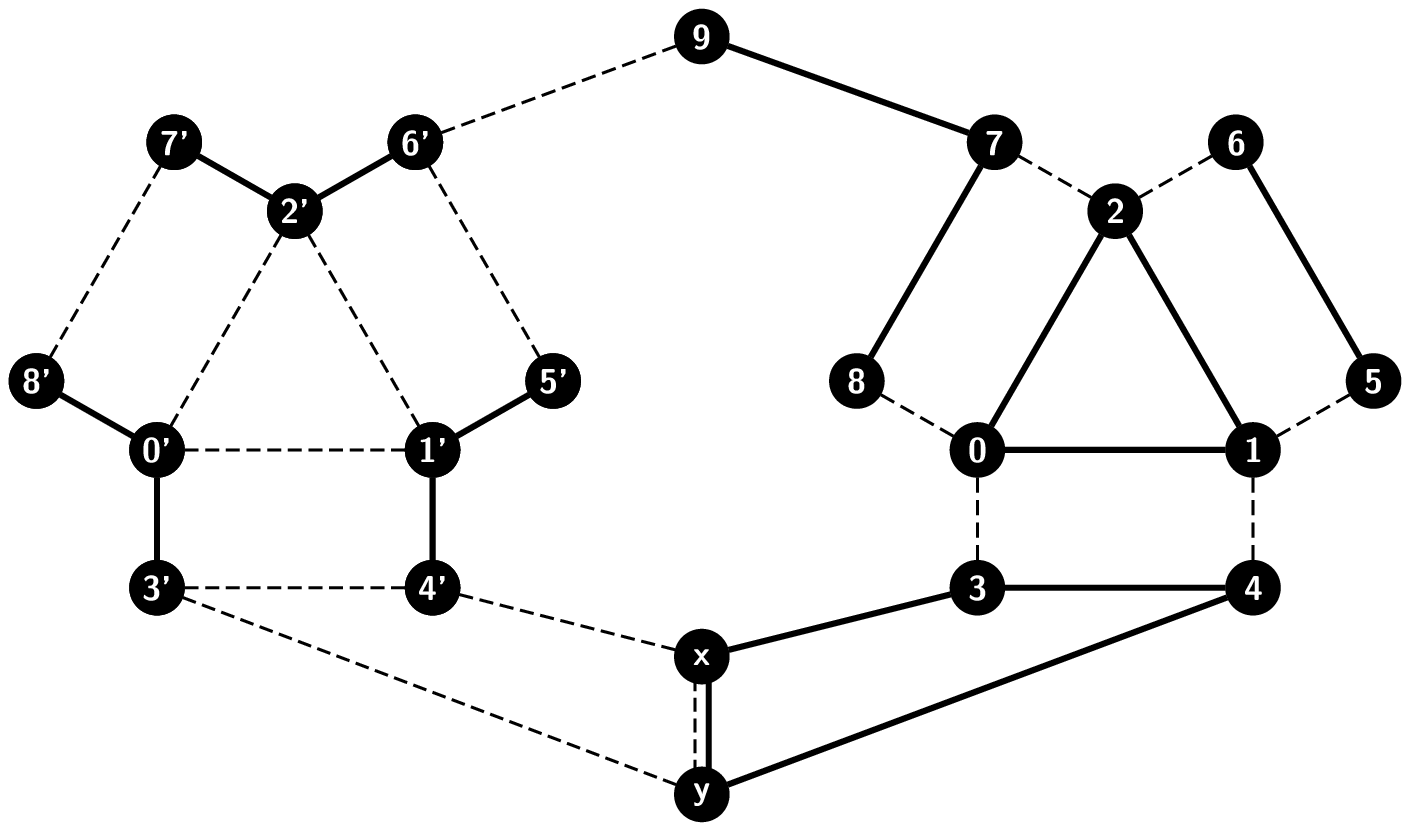}\\[0.3cm]  \centering
  \includegraphics[bb= 72 497 561 609, scale=0.7]{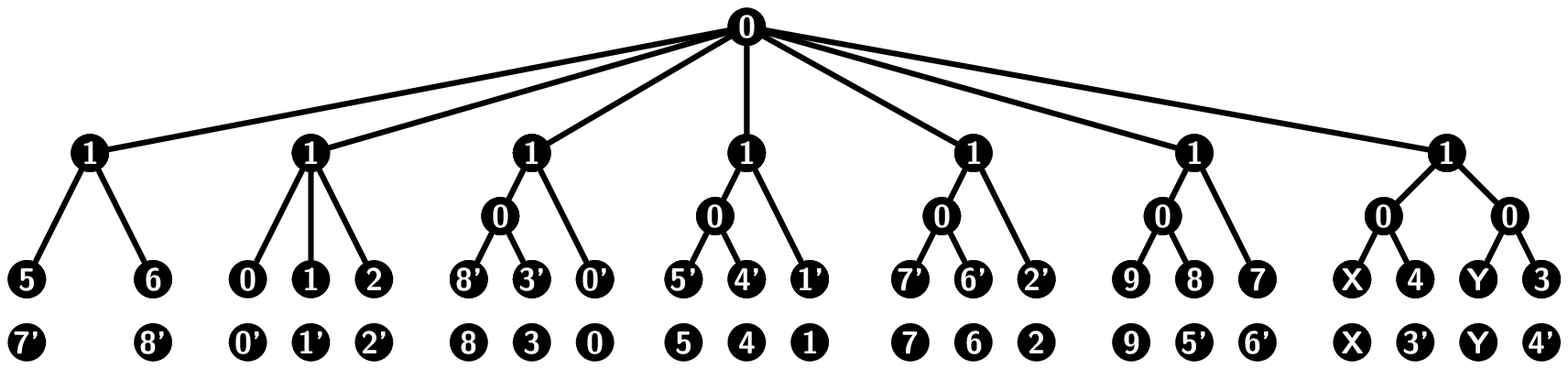}
  \caption{The shown (non-co)graph $G$  has a 2-decomposition $\Pi=\{E_1,E_2\}$.
						Edges in the different elements $E_1$ and $E_2$ are 
            highlighted by dashed and solid edges, respectively. Thus, 
            two cotrees, shown in the lower part of this picture, are sufficient to represent the structure of $G$.
						The two cotrees are isomorphic, and thus, differ only in the arrangement of their leaf sets. For this
		  	   reason, we only depicted one cotree with two different leaf sets. 
					  Note,  $G$ has no 2-partition, but a coarsest 3-partition. The latter can easily be
            verified by application of the construction in Lemma \ref{lem:col-lit}. 
}
  \label{fig:pd}
\end{figure}

However, there are examples of non-cographs containing many induced $P_4$'s that have 
a cograph $k$-decomposition with $k\ll\Delta+1$, which implies that
any optimal $k$-decomposition of those graphs 
will have significantly less elements than $\Delta+1$, see the following
examples. 

\begin{example}
	Consider the graph $G = (V,E)$ with vertex set $V=\{1,\dots,k\}\cup\{a,b\}$
	and $E =\{[i,j]\mid i,j\in \{1,\dots,k\}, i\neq j\}\cup\{[k,a],[a,b]\}$.
	The graph $G$ is not a cograph, since there are induced $P_4$'s of the form
	$i-k-a-b$, $i\in \{1,\dots,k-1\}$. On the other hand, the subgraph
	$H=(V, E\setminus \{[k,a]\})$ has two connected components, 
	one is isomorphic to the complete graph $K_k$ on $k$ vertices and the 	
	other to the complete graph $K_2$. Hence, $H$ is a cograph. 
	Therefore, $G$ has a cograph 2-partition $\{E\setminus \{[k,a]\}, \{[k,a]\}\}$, independent from $k$
	and thus, independent from the maximum degree $\Delta = k$. 
\end{example}

\begin{example}
 Consider the 2n-dimensional hypercube $Q_{2n}=(V,E)$ with maximum degree $2n$. 
 We will show that this hypercube has a
 coarsest cograph $n$-partition $\Pi=\{E_1,\dots,E_n\}$, which implies that for any optimal
 cograph $k$-decomposition of $Q_{2n}$ we have $k\leq \Delta/2$. 

 We construct now a cograph $n$-partition of $Q_{2n}$. Note, $Q_{2n} =
\Box_{i=1}^{2n} K_2 = \Box_{i=1}^n (K_2\Box K_2) = \Box_{i=1}^n Q_2$. In
order to avoid ambiguity, we write $\Box_{i=1}^n Q_2$ as $\Box_{i=1}^n
H_i$, $H_i\simeq Q_2$ 
and assume that $Q_2$ has edges $[0,1]$, $[1,2]$,
$[2,3]$, $[3,0]$. 
The cograph $n$-partition of $Q_{2n}$ is defined as $\Pi
= \{E_1,\dots, E_n\}$, where $E_i=\cup_{v\in V} E(H_i^v)$. In other words,
the edge set of all $H_i$-layers in $Q_{2n}$  
 constitute a single class $E_i$ in
the partition for each $i$. 
Therefore, the subgraph
$G=(V,E_i)$ consists of $n$ connected components, each component is isomorphic
to the square $Q_2$. Hence, $G_i=(V,E_i)$ is a cograph. 

Assume for contradiction that $\Pi = \{E_1,\dots, E_n\}$ is not a coarsest partition.
Then there are distinct classes $E_i$, $i\in I\subseteq \{1,\dots,n\}$ such
that $G_I=(V,\cup_{i\in I}E_i)$ is a cograph. W.l.o.g.\ assume that $1,2\in
I$ and let $v=(0,\dots,0)\in V$. Then, the subgraph $H_1^v\cup H_2^v
\subseteq Q_{2n}$ contains a path $P_4$ with edges $[x,v]\in E(H_1^v)$
and $[v,a], [a,b]\in E(H_2^v)$, where x=(1,0,\dots,0), a=(0,1,0\dots,0)
and $b=(0,2,0\dots,0)$. 
%	W.l.o.g. let $x_1=(1,0,\dots,0)$, $x_2=(0,1,0,\dots,0)$ and $y_2=(0,2,0,\dots,0)$. 
	By definition of the Cartesian product, there are no edges connecting
$x$ with $a$ or $b$ or $v$ with $b$ in $Q_{2n}$ and thus, this path $P_4$ is induced.
As this holds for all subgraphs $H_i^v\cup H_j^v$ ($i,j\in I$ distinct) and
thus, in particular for the graph $G_I$ we can conclude that classes of
$\Pi$ cannot be combined. Hence $\Pi$ is a coarsest cograph $n$-partition. 
% Note, $Q_{2n}$ has ... (exponen?) many induced $P_4$'s.
\label{ex:hq}
\end{example}

Because of the results of computer-aided search for $n-1$-partitions and
decompositions of hypercubes $Q_{2n}$ 
we are led to the following conjecture:

\begin{conjecture}
	Let $k\in \mathbb N$ and $k>1$.
	Then the $2k$-cube has no cograph $k-1$-decomposition, i.e., 
	the proposed $k$-partition of the hypercube $Q_{2k}$  in Example \ref{ex:hq}  is also optimal.
	\label{conj:k-part2}
\end{conjecture}

The proof of the latter hypothesis would immediately verify the
next conjecture. 

\begin{conjecture}
	For every $k\in \mathbb N$ there is a graph that has
	an \emph{optimal} cograph $k$-decomposition. 
%	a  but no cograph $k-1$-partition, i.e., for 
%	this graph $G$ 
%	any $k$-partition is also optimal.
	\label{conj:k-part1}
\end{conjecture}

Proving the last conjecture appears to be difficult.
We wish to point out that there is  a close relationship
to the problem of finding pattern avoiding words, see e.g.\
\cite{Br:05,BM:08,P:08,pudwell2008enumeration,bilotta2013counting,bernini2007enumeration}:
Consider a graph $G=(V,E)$ and an ordered list $(e_1,\dots,e_m)$ of the edges $e_i\in E$. 
We can associate to this list  $(e_1,\dots,e_m)$ a word $w=(w_1,\dots,w_m)$. 
By way of example, assume that we want to find a valid cograph 2-decomposition $\{E_1,E_2\}$ of $G$
and that $G$ contains an induced $P_4$ consisting of the edges $e_i,e_j,e_k$. 
Hence, one has to avoid assignments of the edges $e_i,e_j,e_k$
to the single set $E_1$, resp., $E_2$. The latter is equivalent to
find a binary word $(w_1,\dots,w_m)$  such that 
$(w_i,w_j,w_k) \neq (X,X,X)$, $X\in\{0,1\}$ for each of
those induced $P_4$'s.
The latter can easily be generalized
to find pattern avoiding words over an alphabet $\{1,\dots,k\}$ to get 
a valid $k$-decomposition. 
However, to the authors knowledge, results concerning the counting
of $k$-ary words, avoiding forbidden patterns and thus, verifying if there is any such word
(or equivalently a $k$-decomposition) are basically known for scenarios
like:
If  $(p_1,\dots p_l) \in \{1,\dots,k\}^l$ (often $l<3$), 
then  none of the words $w$ that contain
a subword $(w{_{i_1}},\dots,w{_{i_l}})=(p_1,\dots p_l)$ 
with $i_{j+1}=i_{j}+1$ (consecutive letter positions) or
$i_j<i_k$ whenever $j<k$ (order-isomorphic letter positions)  is allowed. 
However, such findings are to restrictive to our problem, since we are looking
for words, that have only on a few, but fixed positions of non-allowed patterns.
Nevertheless, we assume that results concerning the recognition of 
pattern avoiding words might offer an avenue to solve the latter conjectures.

\subsection{Computational Complexity}

In the following, we will prove the NP-completeness of
\textsc{Cotree 2-Representation} and \textsc{Cotree 2-Decomposition}. 
Additionally, these results allow to show that the problem of determining whether
there is cograph 2-partition is NP-complete, as well. 

We start with two lemmata concerning cograph 2-decompositions of the
graphs shown in Fig.\ \ref{fig:literal} and \ref{fig:clause}.

\begin{lemma}
For the literal and extended literal graph in 
Figure \ref{fig:literal} every cograph 2-decomposition
is a uniquely determined cograph 2-partition.

In particular, in every cograph 2-partition $\{E_1,E_2\}$
of the extended literal graph,
the edges of the triangle $(0,1,2)$ must be entirely
contained in one $E_i$ and the pending edge 
$[6,9]$ must be in the same edge set $E_i$ as the edges of the 
of the triangle. Furthermore, the edges $[9,10]$ and $[9,11]$ must be contained
in $E_j$, $i\neq j$. %, where $[6,9]\notin E_j$
\label{lem:col-lit}
\end{lemma}
\begin{proof}
It is easy to verify that the given cograph 2-partition $\{E_1,E_2\}$ in Fig.
\ref{fig:literal} fulfills the conditions and 
is correct, since $G_1=(V,E_1)$ and $G_2=(V,E_2)$ do not contain
induced $P_4$'s and are, thus, cographs.
We have to show that it is also unique.

Assume that there is another cograph 2-decomposition $\{F_1,F_2\}$. 
Note, for any cograph 2-decomposition $\{F_1,F_2\}$ it must hold that 
two incident edges in the triangle $(0,1,2)$ are contained in 
one of the sets $F_1$ or $F_2$. 
W.l.o.g.\ assume that $[0,1], [0,2]\in F_1$.

Assume first that  $[1,2] \not\in F_1$. 
In this case, because of the paths $P_4 = 6-2-0-1$ and $P_4 = 2-0-1-5$
it must hold that 
$[2,6], [1,5]\not\in F_1$ and thus, $[2,6], [1,5]\in F_2$. 
However, in this case and due to the paths $P_4 = 6-2-1-4$ and $2-0-1-4$ 
the edge $[1,4]$ can neither be contained in $F_1$ nor in $F_2$, 
a contradiction.  Hence, $[1,2] \in F_1$.

Note, the square $S_{1256}$ induced by vertices $1,2,5,6$ cannot have all
edges in $F_1$, as otherwise the subgraph $(V,F_1)$ would contain
the induced $P_4 = 6-5-1-0$. 
Assume that  
$[1,5]\in F_1$. As not all edges $S_{1256}$ are contained
in $F_1$, at least one of the edges $[5,6]$ and $[2,6]$ 
must be contained in $F_2$. If only one of the edges $[5,6]$, resp., $[2,6]$
is contained in $F_2$, we immediately obtain the 
induced $P_4=6-2-1-5$, resp., $6-5-1-2$ in $(V,F_1)$ and therefore, 
both edges $[5,6]$ and $[2,6]$ must be contained in $F_2$. 
But then the edge $[2,7]$ can neither be contained in $F_1$ (due to the induced $P_4=5-1-2-7$)
nor in $F_2$ (due to the induced $P_4=5-6-2-7$), a contradiction. 
Hence, $[1,5]\not\in F_1$ and thus, $[1,5]\in F_2$ for any $2$-decomposition. 
By analogous arguments and due to symmetry, all edges $[0,3]$, $[0,8]$, $[1,4]$, $[2,6]$, $[2,7]$
are contained in $F_2$, but not in $F_1$.

Moreover, due to the induced $P_4= 7-2-6-5$ and since $[2,6], [2,7] \in F_2$, 
the edge $[5,6]$ must be in $F_1$ and not in $F_2$. By analogous arguments and due to symmetry, 
it holds that $[3,4],[7,8]\in F_1$ and  $[3,4],[7,8]\not\in F_2$.
Finally, none of the edges of the triangle $(0,1,2)$ can be contained 
in $F_2$, as otherwise, we obtain an induced $P_4$ in $(V,F_2)$. 
Taken together, any $2$-decomposition of the literal graph must be a 
partition and is unique. 

Consider now the extended literal graph in Figure \ref{fig:literal}. As this
graph contains the literal graph as induced subgraph, the unique $2$-partition
of the underlying literal graph
is determined as by the preceding construction. 
Due to the path $P_4 = 7-2-6-9$ with $[2,6], [2,7] \in F_2$ we can 
conclude that $[6,9]\not\in F_2$ and thus $[6,9]\in F_1$. 
Since there are induced paths $P_4 = 5-6-9-y$, $y=10,11$ 
with $[5,6],[6,9]\in F_1$ we obtain that $[9,10], [9,11]\not\in F_1$
and thus, $[9,10], [9,11]\in F_2$ for any
$2$-decomposition (which is in fact a $2$-partition) 
of the extended literal graph, as claimed. \hfill \qed\end{proof}

\begin{figure}[tbp]
  \centering
  \includegraphics[bb= 0 475 398 660, scale=0.5]{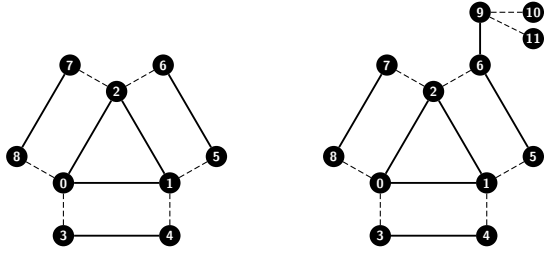}
  \caption{Left the \emph{literal graph} and right the \emph{extended} literal graph with 
			unique corresponding cograph 2-partition (indicated by dashed and bold-lined edges)
			is shown.}
  \label{fig:literal}
\end{figure}

\begin{figure}[tbp]
  \centering
  \includegraphics[bb=  4 322 554 611, scale=0.5]{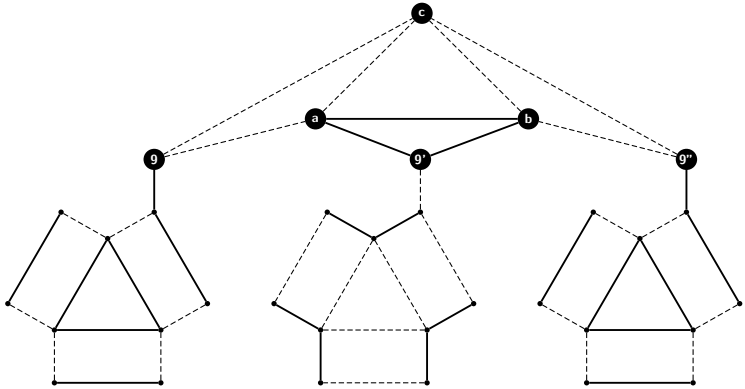}
  \caption{ Shown is a \emph{clause gadget} which consists of a triangle $(a,b,c)$ and
			 three extended literal graphs (as shown in Fig.\ \ref{fig:literal}) with edges attached to $(a,b,c)$. 
			A corresponding cograph 2-partition is indicated by dashed and bold-line edges. 
}
  \label{fig:clause}
\end{figure}

\begin{lemma}
	Given the clause gadget in Fig.\ \ref{fig:clause}. 

	For any cograph 2-decomposition,  
  all edges of exactly two of the triangles in the underlying
	three extended literal graphs must be contained in one $E_i$
	and not in $E_j$,
	while the edges of the triangle of one extended literal graph must be in  $E_j$
	and not in $E_i$, $i\neq j$. 

	Furthermore, 
	for each cograph 2-decomposition exactly two of the edges $e,e'$
	of the triangle $(a,b,c)$ must be in one $E_i$ while the
	other edge $f$ is in $E_j$ but not in $E_i$, $j\neq i$. 
	The cograph 2-decomposition can be chosen so that 
	in addition 
	$e,e'\not \in E_j$, %and $f\not\in E_i$, $j\neq i$, 
	resulting in a cograph 2-partition of the 
	clause gadget.
	\label{lem:col-clause}
\end{lemma}
\begin{proof}
It is easy to verify that the given cograph 2-partition in Fig. \ref{fig:clause} 
	fulfills the conditions and is correct, as $G_1=(V,E_1)$ and $G_2=(V,E_2)$ are cographs.

	As the clause gadget contains the literal graph as induced subgraph, the unique $2$-partition
	of the underlying literal graph is determined as by the construction given in Lemma \ref{lem:col-lit}. 
    Thus, each edge of the triangle in each underlying literal graph is contained in either one 
	of the sets $E_1$ or $E_2$. 
	Assume that edges of the triangles in the three literal gadgets are \emph{all} contained
	in the same set, say $E_1$. 
	Then, Lemma \ref{lem:col-lit} implies that 
	$[9,a], [9,c], [9',a],[9',b],[9'',b], [9'',c] \in E_2$
	and none of them is contained in $E_1$. 
	Since there are induced $P_4$'s: $9-a-b-9''$, $9'-a-c-9''$ and $9-c-b-9'$, 
	the edges 
	$[a,b], [a,c], [b,c]$ cannot be contained in $E_2$, and thus must be
	in $E_1$. However, this is 
	not possible, since then we would have the induced paths $P_4=9-a-9'-b$ in 
	the subgraph $(V,E_2)$
	a contradiction. Thus, the edges of the triangle of exactly one literal gadget 
    must be contained in a different set $E_i$ than
	the edges of the other triangles in the other two literal gadgets. 
	W.l.o.g. assume that the $2$-decomposition of the underlying 
	literal gadgets is given as in Fig.\ \ref{fig:clause}. 
	and identify bold-lined edges with $E_1$ and dashed edges with $E_2$. 
	
    It remains to show that this 2-decomposition of the underlying three 
	literal gadgets
	 determines %\TODO{in a unique way} 
	which of the edges
	of triangle $(a,b,c)$ are contained in which of the sets $E_1$ and $E_2$.	
	Due to the induced path $9-a-b-9''$ and since $[9,a],[9'',b]\in E_2$, the edge $[a,b]$ cannot be contained
	in $E_2$ and thus, is contained in $E_1$. 
	Moreover, if $[b,c]\not\in E_2$, then  
    there is an induced path $P_4 = b-9''-c-9$ in the subgraph $(V,E_2)$, a
	contradiction. Hence,  $[b,c] \in E_2$ and by analogous arguments, $[a,c]\in E_2$. 
	If  $[b,c] \not\in E_1$ and  $[a,c] \not\in E_1$, then we
	obtain a cograph 2-partition. However, it can easily be verified that 
	there is still a degree of 
	freedom and $[a,c], [b,c]\in E_1$ is allowed for a 
	valid cograph 2-decomposition. \hfill \qed
\end{proof}

\begin{figure}[tbp]
  \centering
  \includegraphics[bb=  9 490 607 679, scale=0.55]{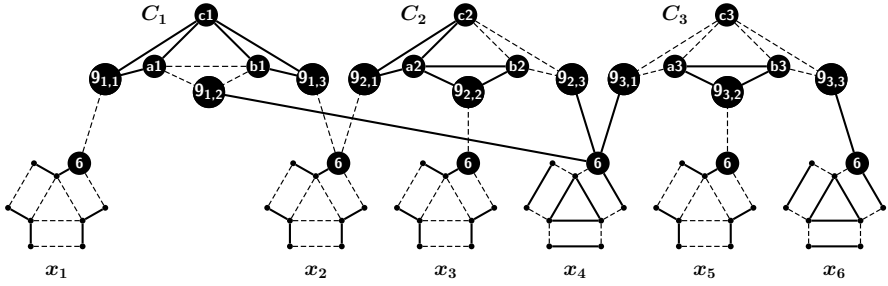}
  \caption{Shown is the graph $\Psi$ as constructed in the proof of Theorem \ref{thm:npc}.
			In particular, $\Psi$ reflects the NAE 3-SAT formula 
  			$\psi = \{C_1, C_2, C_3\}$ with clauses $C_1=(x_1,x_4,x_2), C_2=(x_2,x_3,x_4)$ and 
  				$C_3=(x_4,x_5,x_6)$. Different literals obtain the same truth assignment true or false, 
				whenever the edges of the triangle in their corresponding literal gadget are contained in 
				the same set $E_i$ of the cograph 2-partition, highlighted by dashed and bold-lined edges.  }
  \label{fig:formula}
\end{figure}

We are now in the position  to prove the NP-completeness of
\textsc{Cotree 2-Representation} and \textsc{Cotree 2-Decomposition}
by reduction from the following problem. 

\begin{problem1}\textsc{Monotone NAE 3-SAT}\\[0.1cm]
\begin{tabular}{ll}
	  \emph{Input:}&Given a set $U$ of Boolean variables and a set of clauses \\
					&$\psi = \{C_1, \dots, C_m\}$ over $U$
					such that for all $i=1, \dots, m$ \\
					&it holds that $|C_i|=3$ and $C_i$ contains no negated variables. \\
	\emph{Question:}&Is there a truth assignment to $\psi$ such that in each $C_i$\\
					&not all three literals are set to true?
\end{tabular}
\end{problem1}

\begin{theorem}[\cite{Sch78,moret-97}]
\textsc{Monotone NAE 3-SAT} is NP-complete. 
\end{theorem}

%paper: On the complexity of the balanced vertex ordering problem
%http://pages.cs.wisc.edu/~tdw/files/slides/2013-05-17_Google-Madison.pdf
%http://tuvalu.santafe.edu/~moore/theory/hw4solns.pdf

% mon NAE 3 sat from nae-3 sat by replacing all literals
% of the form 1-v_i by y_i and adding additional clause 
% y_i,v_i,v_i (http://revealedpreferences.org/articles/1310635214main.pdf)

%http://www.cs.mcgill.ca/~sue/506/announcements/Schaefer.pdf

\begin{theorem}
	\textsc{Cograph 2-Decomposition}, and thus, \textsc{Cotree 2-Representation} is NP-complete.
	\label{thm:npc}
\end{theorem}
\begin{proof}
Given a graph $G=(V,E)$ and cograph 2-decomposition $\{E_1,E_2\}$, 
	one can verify in linear time whether $(V,E_i)$ is a cograph \cite{corneil1985linear}.
	Hence, \textsc{Cograph 2-Partition} $\in$ NP. 

   We will show by reduction from \textsc{Monotone NAE 3-SAT} that
   \textsc{Cograph 2-Decomposition} is NP-hard. %, even for $k=2$. 
   Let $\psi = (C_1, \dots, C_m)$ be an arbitrary instance of
	\textsc{Monotone NAE 3-SAT}.
	Each clause $C_i$ is identified with a triangle $(a_i,b_i,c_i)$.
	Each variable $x_j$ is identified with a literal graph 
	as shown in Fig.\ \ref{fig:literal} (left) and different
	variables are identified with different literal graphs. 
	Let $C_i = (x_{i_1}, x_{i_2}, x_{i_3})$ and $G_{i_1}$,
	$G_{i_2}$ and $G_{i_3}$ the respective literal graphs. 
	Then, we extend each
	literal graph $G_{i_j}$ by adding an edge $[6,9_{i,j}]$. 
	Moreover, we add to $G_{i_1}$ the edges $[9_{i,1},a_i], [9_{i,1},c_i]$, 
	to $G_{i_2}$ the edges $[9_{i,2},a_i], [9_{i,2},b_i]$, 
	to $G_{i_3}$ the edges $[9_{i,3},c_i], [9_{i,3},b_i]$. The latter construction
	connects each literal graph with the triangle $(a_i,b_i,c_i)$ of
	the respective clause $C_i$ in a unique way, see Fig. \ref{fig:clause}.
	We denote the clause gadgets
	by $\Psi_i$ for each clause $C_i$. We repeat this construction
	for all clauses $C_i$ of $\psi$ resulting in the graph  $\Psi$. 
	An illustrative example is given in   Fig.\ \ref{fig:formula}.	
	Clearly, this reduction can be done in polynomial time
	in the number $m$ of clauses.  

	We will show in the following that $\Psi$ has a 
	cograph 2-decomposition (resp., a cograph 2-partition) 
	if and only if $\psi$ has a truth assignment $f$. 

	Let $\psi = (C_1, \dots, C_m)$ have a truth assignment. 
	Then in each clause $C_i$ at least one of the literals $x_{i_1}, x_{i_2}, x_{i_3}$ 
	is set to true and one to false. 
	We assign all edges $e$ of the triangle in the 
	corresponding literal graph $G_{i_j}$ to $E_1$, 
	if 	$f(x_{i_j})=true$ and to $E_2$, otherwise.
	Hence, each edge of exactly two of the triangles (one in $G_{i_j}$ and one in $G_{i_{j'}}$) are 
	contained in one $E_r$	and not in $E_s$, 
	while the edges of the other triangle in $G_{i_{j''}}$, $j''\neq j,j'$
	are contained in  $E_s$	and not in $E_r$, $r\neq s$, 
	as needed
	for a possible valid cograph 2-decomposition (Lemma \ref{lem:col-clause}). 
	We now apply the construction of a valid cograph 2-decomposition (or cograph 2-partition) 
	for each $\Psi_i$ as given in Lemma \ref{lem:col-clause}, 
	starting with the just created assignment of edges contained in 	
	the triangles in $G_{i_j}$, $G_{i_{j'}}$ and $G_{i_{j''}}$ to $E_1$ or $E_2$. In this way, we obtain a valid 
	cograph 2-decomposition (or cograph 2-partition) for each subgraph $\Psi_i$ of $\Psi$.
	Thus, if there would be an induced $P_4$ in $\Psi$ with all edges 
	belonging to the same set $E_r$, 
	then this $P_4$ can only have edges belonging to different 
	clause gadgets $\Psi_k, \Psi_l$. 
	By construction, such a $P_4$ can only exist along 
	different clause gadgets $\Psi_k$ and $\Psi_l$ 
	if $C_k$ and $C_l$  have a literal $x_i=x_{k_m}=x_{l_n}$ in 
	common. In this case, Lemma \ref{lem:col-clause} implies that 
	the edges $[6,9_{k,m}]$ and $[6,9_{l,n}]$ in $\Psi_i$ must belong
	to the same set $E_r$. % and obtain therefore the same color. 
	Again by Lemma \ref{lem:col-clause}, the edges $[9_{k,m},y]$
	and  $[9_{k,m},y']$,\ $y,y'\in \{a_k,b_k,c_k\}$ as well as
	the edges $[9_{l,n},y]$
	and  $[9_{l,n},y']$,\ $y,y'\in \{a_l,b_l,c_l\}$ must be
	in a different set $E_s$ than $[6,9_{k,m}]$ and $[6,9_{l,n}]$. 
	Moreover, respective edges $[5,6]$ in $\Psi_k$, as well as in 
	 $\Psi_l$ (Fig.\ \ref{fig:literal}) must 
	be in $E_r$, i.e., in the same set as $[6,9_{k,m}]$ and $[6,9_{l,n}]$. 
	However, in none of the cases it is possible to find
	an induced $P_4$ with all edges 
	in the same set $E_r$ or $E_s$ along different clause gadgets. 
	Hence, we 
	obtain a valid cograph 2-decomposition, resp., cograph 2-partition of $\Psi$. 

	Now assume that $\Psi$ has a valid cograph 2-decomposition (or a
	cograph 2-partition). Any variable $x_{i_j}$ contained in some clause $C_i =
	(x_{i_1}, x_{i_2}, x_{i_3})$ is identified with a literal graph
	$G_{i_j}$. Each clause $C_i$ is, by construction, identified with
	exactly three literal graphs $G_{i_1},G_{i_2},G_{i_3}$, resulting in
	the clause gadget $\Psi_i$. Each literal graph $G_{i_j}$ contains
	exactly one triangle $t_j$. Since $\Psi_i$ is an induced subgraph of
	$\Psi$, we can apply Lemma \ref{lem:col-clause} and conclude that for
	any cograph 2-decomposition (resp., cograph 2-partition) all edges of exactly
	two of three triangles $t_1,t_2,t_3$ are contained in one set $E_r$,
	but not in $E_s$, and all edges of the other triangle are contained in
	$E_s$, but not in $E_r$, $s\neq r$. Based on these triangles we define
	a truth assignment $f$ to the corresponding literals: w.l.o.g.\ we
	set $f(x_i)=$\emph{true} if the edge $e\in t_i$ is contained in $E_1$
	and $f(x_i)=$\emph{false} otherwise. By the latter arguments and Lemma
	\ref{lem:col-clause}, we can conclude that, given a valid cograph
	2-partitioning, the so defined truth assignment $f$ is  a valid truth
	assignment of the Boolean formula $\psi$, since no three different literals in one
	clause obtain the same assignment and at least one of the variables is
	set to $\emph{true}$. Thus, \textsc{Cograph 2-Decomposition} is
	NP-complete
	
	Finally, because \textsc{Cograph 2-Decomposition} and \textsc{Cotree 2-Representation} 
	are equivalent problems, 
	the NP-completeness of \textsc{Cotree 2-Representation} follows.  \hfill \qed
\end{proof}

As the proof of Theorem \ref{thm:npc} allows us to use cograph 2-partitions in all proof steps,
instead of cograph 2-decompositions, we can immediately infer the
NP-completeness of the following problem for $k=2$, as well.

\begin{problem1}\textsc{Cograph $k$-Partition}\\[0.1cm]
\begin{tabular}{ll}
	 \emph{Input:}&Given a graph $G=(V,E)$ and an integer $k$. \\
	 \emph{Question:}&Is there a Cograph $k$-Partition of $G$?
\end{tabular}
\end{problem1}

\begin{theorem}
	\textsc{Cograph 2-Partition} is NP-complete. 
\end{theorem}

As a direct consequence of the latter results, we obtain the 
following result.

\begin{corollary}
Let $G$ be a given graph that is not a cograph. 
The  three optimization problems to find the least integer $k>1$ so that 
 there is 
%	\begin{center}
	a Cograph $k$-Partition,  
	a Cograph $k$-Decomposition, or
	a Cotree $k$-Representation
%	\end{center}
%\begin{itemize}
%	\item[$\bullet$] a Cograph $k$-Partition,  
%	\item[$\bullet$] a Cograph $k$-Decomposition, and
%	\item[$\bullet$] a Cotree $k$-Representation
%\end{itemize}
for the graph $G$, are NP-hard. 
\end{corollary}

%\TODO{a few words:if partition so that triangles correct then we get umetric. also with complemental-edges}

\subsection{Integer Linear Program}

Let $G=(V,E)$ be a given graph with maximum degree $\Delta$. We want to
find a cograph-$k$-decomposition, resp., partition $\Pi= \{E_1,\dots,E_k\}$
for the least integer $k$. Theorem \ref{thm:k-delta} implies that
the least integer $k$ is always less or equal to $\Delta+1$.

We define binary variables $E^i_{xy}$ for all $x,y\in V$ and 
$1\leq i\leq \Delta+1$ s.t. $E^i_{xy}=1$ if and only if
the edge $[x,y]\in E$ is contained in class $E_i$ of $\Pi$.
Moreover, we define the binary variables $M^i$ with $1 \leq i \leq \Delta+1$ 
so that $M^i=1$ if and only if the class $E_i\in\Pi$ is non-empty
in our construction. In other words, 
$\sum_{1 \leq i \leq \Delta+1} M^i$ will be the cardinality of $\Pi$.

In order to find the cograph decomposition, resp., partition $\Pi$ of $G$ 
having the fewest number of elements we need the following objective
function.
\begin{equation}
\min \sum_{1 \leq i \leq \Delta+1} M^i
\end{equation}

If we want to find a cograph-decomposition and
hence, that each edge is contained in at least one class $E_i$ of $\Pi$
we need the next constraint.
\begin{align}
\sum_{1 \leq i \leq \Delta+1} E^i_{xy} \geq 1   \text{ for all } [x,y] \in E.
\label{dDDD}
\end{align}

In contrast, if we want to find a cograph-partition and
hence, that each edge is contained in exactly one class $E_i$ of $\Pi$
we need this constraint.
\begin{align}
\sum_{1 \leq i \leq \Delta+1} E^i_{xy} = 1  \text{ for all } [x,y] \in E. \tag{\ref{dDDD}'}
\end{align}

Moreover, we must ensure that non-edges $[x,y]\notin E$ are not contained in 
any class of $\Pi$ which is done with the next constraint.
\begin{equation}
\sum_{1 \leq i \leq \Delta+1} E^i_{xy} = 0 \text{ for all } [x,y] \notin E. 
\end{equation}

Whenever there is a class $E_i$ containing an edge $[x,y]\in E$
and hence, if  $E^i_{xy}=1$ then we must set 
$M^i=1$. 
\begin{equation}
\sum_{x,y \in V} E^i_{xy} \leq |V|^2 M^i \text{ for all } 1 \leq i \leq \Delta+1.
\end{equation}

Finally we have to ensure that each subgraph $G_i=(V,E_i)$ is a cograph,
and thus, does not contain induced $P_4$'s,
which is achieved with the following constraint.
\begin{equation}
E^i_{xy} + E^i_{yu} + E^i_{uv} - E^i_{xu} - E^i_{xv} - E^i_{yv} \leq 2
\end{equation}
for all $ 1 \leq i \leq \Delta+1$ and all ordered tuples $(x,y,u,v)$ of distinct $x,y,u,v \in V$.

This ILP-formulation needs $O(\Delta |V|^2)$ variables
and $O(|E|+\Delta+|V|^4)$ constraints.

%\section{Outlook and Summary}

%bla 

%A further problems arising from the latter considerations are as follows. 
%Assume we have an assignment of a symmetric subset $R$ of $X\times X$ so that for
%all $(x,y)\in R$ the assignment $d(x,y)$ is believed to be an reliable estimate 
%and thus, which is not allowed to be changed. 
%Moreover, let $X\times X\setminus R$
%be the pairs $(x,y)$ for which an assignment $d(x,y)$ is not known.
%%i.e., they might correspond to any assignment $m\in \Mo$. 
%If  only non-reliable estimates of pairs $(x,y)$ are allowed to be changed.
%We might ask for the change of the assignment of \emph{all} pairs $(x,y)$ 
%with $d(x,y)=0$ to some element in $m\in M, m\neq0$ 
%so that the resulting map is a symbolic ultrametric. 

%\begin{problem1}\textsc{Symbolic Ultrametric Full Completion}

%\vspace{0.3em}
%\begin{tabular}{ll}
%	\emph{Input:}& Given a symmetric map $d : X \times X \to \Mo$ with \\ 
%					 &	$M=\{0,\ldots,n\}$, s.t.\ %\\
%					  $d(x,y)=\varnothing$ if and only if $x=y$.\\
%	\emph{Question:}& Is there a symbolic ultrametric $\delta : X \times X \to \Mo\setminus \{0\}$, s.t.\  \\
% 			 			 & $\delta(x,y)=d(x,y)$ if $d(x,y)\neq 0$ and $\delta(x,y)\in \Mo$ if $d(x,y) = 0$
%						
%\end{tabular}
%\end{problem1}

%\begin{conjecture}
%	\textsc{Symbolic Ultrametric Full Completion} is NP-complete.
%\end{conjecture}

%\begin{acknowledgements}
%If you'd like to thank anyone, place your comments here
%and remove the percent signs.
%\end{acknowledgements}

% BibTeX users please use one of
%\bibliographystyle{spbasic}      % basic style, author-year citations
\bibliographystyle{spmpsci}      % mathematics and physical sciences
\bibliography{paper}
%\bibliography{}   % name your BibTeX data base

% Non-BibTeX users please use
%\begin{thebibliography}{}
%%
%% and use \bibitem to create references. Consult the Instructions
%% for authors for reference list style.
%%
%\bibitem{RefJ}
%% Format for Journal Reference
%Author, Article title, Journal, Volume, page numbers (year)
%% Format for books
%\bibitem{RefB}
%Author, Book title, page numbers. Publisher, place (year)
%% etc
%\end{thebibliography}

\end{document}